\documentclass[reqno, 11pt, a4paper]{amsart}
\usepackage[utf8]{inputenc}
\usepackage[english]{babel}
\usepackage{latexsym, ifthen, xspace}
\usepackage{amsmath, amssymb, amsthm, thmtools, mathtools}
\usepackage{enumerate, calc}
\usepackage[colorlinks=true, pdfstartview=FitV, linkcolor=blue,
  citecolor=blue, urlcolor=blue, pagebackref=false]{hyperref}
\usepackage[text={40pc,650pt},centering]{geometry}    %11pt

\usepackage[notref, notcite]{}
\usepackage{comment}
\usepackage{graphicx}
\usepackage{subfigure}

\usepackage{csquotes}
\usepackage{amsfonts}

\usepackage{pgfplots}
\pgfplotsset{/pgf/number format/use comma,compat=newest}
\usepackage{url}
\usepackage{tikz}

\usepackage{bbm}
\usepackage{scalerel,stackengine}
\stackMath
\newcommand\reallywidehat[1]{%
\savestack{\tmpbox}{\stretchto{%
  \scaleto{%
    \scalerel*[\widthof{\ensuremath{#1}}]{\kern.1pt\mathchar"0362\kern.1pt}%
    {\rule{0ex}{\textheight}}%WIDTH-LIMITED CIRCUMFLEX
  }{\textheight}% 
}{2.4ex}}%
\stackon[-6.9pt]{#1}{\tmpbox}%
}

\usepackage[most]{tcolorbox}
\tcbset{
        highlight math style= {enhanced, %<-- needed for the ’remember’ options
colframe=red,colback=red!10!white,boxsep=0pt}
        }
\usepackage{floatflt}

\usetikzlibrary{arrows,decorations.pathmorphing,backgrounds,positioning,fit,petri,patterns}
\usepgfplotslibrary{fillbetween}

\tikzstyle{spring}=[thick,decorate,decoration={zigzag,pre length=0.1cm,post
  length=0.1cm,segment length=6}]
\theoremstyle{plain} 

\newtheorem{thm}{Theorem}[section]

\newtheorem{prop}[thm]{Proposition} 

\theoremstyle{definition} 
\newtheorem{defn}{Definition}[section]

\newtheorem{lemma}{Lemma}[section]

\theoremstyle{remark} 
\newtheorem{rem}{Remark}[section]

\DeclareMathOperator\arctanh{arctanh}

\numberwithin{equation}{section}
\usepackage{appendix}
\DeclareMathAlphabet{\mathsl}{OT1}{cmss}{m}{sl}
\SetMathAlphabet{\mathsl}{bold}{OT1}{cmss}{bx}{sl}

\newcommand{\s}{\ensuremath{\sigma}}

\begin{document}
\title[]{Limit theorems for the non-convex multispecies Curie-Weiss model}
% or Will the AI revolution be the first industrial revolution to keep energetic efficiency increasing?

%
%    Remove any unused author tags.
%

%    author one information
\author[]{Francesco Camilli$^1$}
\address{Francesco Camilli}
\curraddr{}
\email{fcamilli@ictp.it}
\thanks{}

%    author two information
\author[]{Emanuele Mingione$^2$}
\address{Emanuele Mingione}
\curraddr{}
\email{emanuele.mingione2@unibo.it}
\thanks{}

%    author three information
\author[]{Godwin Osabutey$^2$\\ \\ \\ \tiny{$^1$Quantitative Life Sciences, International Centre for Theoretical
Physics, Trieste, Italy \\$^2$Department of Matematics, Alma Mater Studiorum - University of Bologna, Bologna Italy}
%\\$^3$ Department of Physics, Computer Science and Mathematics, University of Modena and Reggio Emilia, Modena, Italy
}
\address{Godwin Osabutey}
\curraddr{Department of Physics, Computer Science and Mathematics, University of Modena and Reggio Emilia, Modena, Italy}
\email{gosabutey@unimore.it}
\thanks{}

\subjclass[2000]{}


\date{\today}

\dedicatory{}

\begin{abstract}
We study the thermodynamic properties of the generalized non-convex multispecies Curie-Weiss model, where interactions among different types of particles (forming the species) are encoded in a generic matrix. For spins with a generic prior distribution, we compute the pressure in the thermodynamic limit using simple interpolation techniques. For Ising spins, we further analyze the fluctuations of the magnetization in the thermodynamic limit under the Boltzmann-Gibbs measure. It is shown that a central limit theorem holds for a rescaled and centered vector of species magnetizations, which converges to either a centered or non-centered multivariate normal distribution, depending on the rate of convergence of the relative sizes of the species.
\end{abstract}

\maketitle

\section{Introduction}

The Curie-Weiss model, also known as the mean-field Ising model, is one of the simplest models of magnetism that exhibits phase transitions \cite{Kac}. In this model, the spins  assume discrete binary values ($\pm 1$) and interact uniformly with one another. Due to its simplicity and analytical tractability, it has been applied in a variety of fields, including voting dynamics \cite{Kirsch_Langner_2014, Kirsch_2016, Gsänger_Hösel_MKM_2024} and social collective behavior \cite{Brock_Durlauf_2001, Blume_Durlauf_2003, Gsänger_Hösel_MKM_2024, Marsman_Tanis_BW_2019}. The multispecies extension of the Curie-Weiss model \cite{Contucci_Ghirlanda_2007, BurioniContucciFedeleVerniaVezzani2015, OpokuOsabuteyK2019, ContucciKO2022} has been proposed to capture the large-scale behavior of interacting systems involving multiple types of interacting particles, whose strength depends on which species  each belong to. %These extensions were originally introduced in statistical physics  as approximations of lattice models \cite{FroyenSH1976, Bidaux1986} and meta-magnets \cite{KinCohen1975, GalamYS1998} that display both ferromagnetic and antiferromagnetic interactions,  it turns out that the equilibrium measure in the mean-field approximation is  described by a multispecies Curie-Weiss model. 
These extensions were originally introduced in statistical physics as approximations of lattice models \cite{FroyenSH1976, Bidaux1986} and meta-magnets \cite{KinCohen1975, GalamYS1998}, which exhibit both ferromagnetic and antiferromagnetic interactions. From the mathematical physics perspective the free energy of a two-species ferromagnet was rigorously derived in \cite{Gallo2008} and further investigated in \cite{Genovese_Barra_2009, Barra_Genovese_Guerra_2011, G_Tantari_2016} with the Hamilton-Jacobi formalism. Beyond the computation of the free energy, the fluctuations of the order parameter of the multispecies Curie-Weiss model were initially discussed in \cite{Fedele_Contucci_2011}, under a  convexity assumption on the Hamiltonian  and using the same approach of \cite{Ellis_Newman_1978}. 
More recently stronger results have been proved using different approaches. Notably \cite{Löwe_Schubert_2018, Knöpfel_Löwe_Schubert_Sinulis_2020} have shown validity of central limit theorems (CLTs) and provided their convergence rate via Stein's method, while in \cite{Kirsch_Toth_2020,Fleermann_Kirsch_Toth_2022, Kirsch_Toth_2022} a moment generating functional approach has been used. We stress that in all the aforementioned literature the convexity assumption on the interaction matrix always plays a crucial role. An analogous convexity condition also appears in the context of disordered multispecies models where a formula for the free energy  is  known  only in the convex case  \cite{MSK_original,MSK_Panchenko}, for spherical models \cite{baik2020,Subag} and on the Nishimori line 
\cite{MSKNL, DBMNL, Nishiletter}. The non-convex case remains an open problem (see \cite{noi_deep2, mourrat200,Bimourrat} for partial results on the subject). Concerning models with random interactions, few fluctuation theorems for the order parameter are known: for the Sherrington-Kirkpatrick model the validity of CLTs is limited to high-temperature regions \cite{guerra2002central,guerra2002quadratic,Talagrand2003spin}, unless the model is on the Nishimori line \cite{camilli2023central}, while the multispecies case is studied here \cite{Dey}.

The aim of this work is to study the multispecies Curie Weiss model in full generality, in particular without any convexity assumption on the interaction matrix. 
The work is divided into two main parts. The first part examines the limiting free energy of the model for  arbitrary spin distributions supported on $[-1,1]$. Using a combination of interpolation methods and decoupling techniques, we derive a variational formula for the free energy. The second part  focuses on the asymptotic behavior of the vector of species magnetization in the case of Ising spins. The methods used in the latter are inspired by \cite{ MSB21, Contucci_Mingione_Osabutey_2024}. By generalizing these methods we demonstrate that the rescaled vector of species magnetization follows a standard CLT in the region of phase space where the order parameter concentrates on a single value. On the other hand, when concentration may occur at multiple points, a conditional CLT still applies.

The paper is structured as follows: Section 2 gives a description for the generalized multispecies Curie-Weiss model with a generic coupling matrix. The main results are presented in Section 3, followed by detailed proofs in Section 4. Section 5 concludes the paper and discusses future directions. The Appendix contains some technical results used in the proofs.

\section{Model Description and Definitions}

\tikzset{external/figure name/.add={}{MCW}}
Let $\rho$ be a probability measure supported on $[-1,1]$, and $K$ an integer representing the number of different species. Consider now a Hamiltonian system made of $N$ interacting spins, each labeled by an integer, $\sigma_i$, that lies in the set of indices $\Lambda=\{1,2,\dots,N\}$. To create a multispecies structure we divide the set of indices $\Lambda$ into $K$ disjoint subsets $\Lambda_{p}$ of cardinality $N_p \; \text{for} \; p\,=\,1,\ldots,K$, namely
\begin{equation}\label{densities}
\Lambda_{p}\cap\Lambda_{l}=\emptyset\;\; \forall\,p\neq l, \;  \; \sum_{p=1}^K |\Lambda_{p}|= N_1 + \cdots + N_K = N\,.
\end{equation}
For future convenience we also introduce the \emph{form factors}, or relative sizes ratios $(\alpha_p)_{p\leq K}=(N_p/N)_{p\leq K}$, that shall be collected into a diagonal $K\times K$ matrix $\boldsymbol{\alpha}=\text{diag}(\alpha_p)_{p\leq K}$. Naturally, one has $\sum_{p\leq K}{\alpha_p}=1$. In the following we also allow these ratios to depend on $N$, $\alpha_{N,p}$, and we shall call their limit $\alpha_p$:
\begin{equation}\label{alphaconvergence}
\lim_{N\to\infty }\alpha_{N,p}=\lim_{N\to\infty }\frac{N_p}{N}= \alpha_p ,\;\;p\leq K.
\end{equation}
However, with a little abuse of notation we keep writing only $\alpha_p$ instead of $\alpha_{N,p}$. 

Let us turn back to our Hamiltonian system. The configuration space is $[-1,1]^N=\Omega_N$. Now that the multispecies structure has been introduced, let the interaction be tuned by the symmetric matrix $\mathbf{J} \in \mathbb{R}^{K\times K}$ and the species specific external field $\mathbf{h}\in\mathbb{R}^K$. A schematic representation of the interaction network is displayed in Figure \ref{interaction scheme}.

\begin{figure}[h]
\centering
\begin{tikzpicture}[
roundnode/.style={circle, draw=green!60, fill=green!5, very thick, minimum size=7mm},
squarednode/.style={rectangle, draw=red!60, fill=red!5, very thick, minimum size=4mm},scale=1.2]
\filldraw[color=blue,fill=blue!10] (0,2) circle (1.25);
\filldraw[color=red,fill=red!10] (-2,-1) circle (1.25);
\filldraw[color=green,fill=green!10] (2,-1) circle (1.25);

\filldraw [gray] (-2.2,-1.8) circle (2pt);
\filldraw [gray] (-2.5,-0.2) circle (2pt);
\draw (-2.2,-1.8)--(-2.5,-0.2);

\filldraw [gray] (-0.5,1.5) circle (2pt);
\filldraw [gray] (0.5,2) circle (2pt);
\draw (-0.5,1.5)--(0.5,2);

\filldraw [gray] (2.2,-1.8) circle (2pt);
\filldraw [gray] (2.5,-0.2) circle (2pt);
\draw (2.2,-1.8)--(2.5,-0.2);

\draw (-2.2,-1.8)--(-0.5,1.5);
\draw (-2.2,-1.8)--(0.5,2);
\draw (-2.2,-1.8)--(2.2,-1.8);
\draw (-2.2,-1.8)--(2.5,-0.2);

\draw (-2.5,-0.2)--(-0.5,1.5) ;
\draw (-2.5,-0.2)--(0.5,2) ;
\draw (-2.5,-0.2)--(2.2,-1.8) ;
\draw (-2.5,-0.2)--(2.5,-0.2) ;

\draw (-0.5,1.5)--(2.2,-1.8) ;
\draw (-0.5,1.5)--(2.5,-0.2) ;

\draw (0.5,2)--(2.2,-1.8) ;
\draw (0.5,2)--(2.5,-0.2) ;

\node[align=left] at (1.75,2) {$\Lambda_{p}$};
\node[align=right] at (-3.75,-1) {$\Lambda_{l}$};
\node[align=left] at (3.75,-1) {$\Lambda_{r}$};

\node[align=left] at (-1.5,1) {$J_{pl}$};
\node[align=right] at (0,-2.25) {$J_{lr}$};
\node[align=left] at (1.75,1) {$J_{pr}$};
\end{tikzpicture}
\caption{Interaction scheme for the multispecies model.}\label{interaction scheme} 
\end{figure}
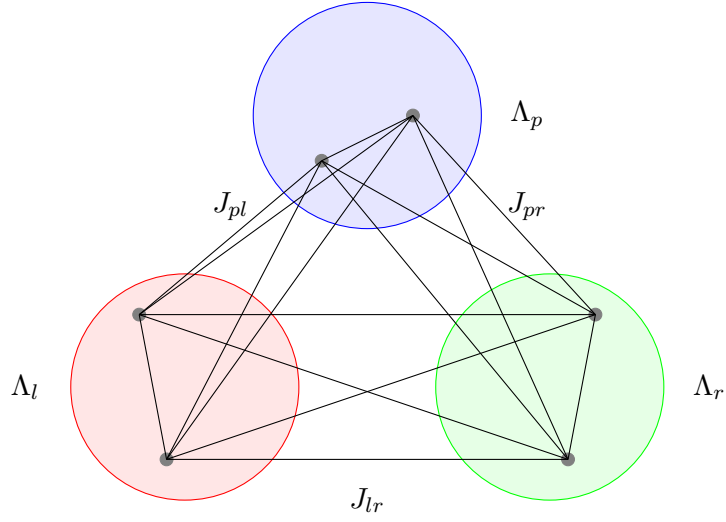

The model under study is thus defined as follows. Let $\sigma=(\sigma_i)_{i\in\Lambda}\in\Omega_N$, then for a given $(\mathbf{J},\mathbf{h})\in \mathbb{R}^{K\times K}\times\mathbb{R}^K $ the multispecies model is defined by the Hamiltonian:

\begin{align}\label{MCW_hamiltonian}
H_N(\sigma)=-\frac{1}{2N}\sum_{i,j=1}^NJ_{ij}\sigma_i\sigma_j-\sum_{i=1}^Nh_i\sigma_i=-\frac{N}{2}\sum_{p,l=1}^K m_p  \alpha_p J_{pl}\alpha_lm_{l}-N\sum_{p=1}^K\alpha_ph_pm_p  
\end{align}
where for each species $p$ the  magnetization density is:
\begin{align}\label{mag ms}
    m_p (\sigma)=\frac{1}{N_p}\sum_{i\in \Lambda_p}\sigma_i \,
\end{align}
and we denote by $\mathbf{m}_N=(m_p )_{p\leq K}\in[-1,1]^K$ the magnetization vector.
The Hamiltonian \eqref{MCW_hamiltonian} can now be rewritten as:
\begin{align}\label{MCW_hamiltonian_matrix}
    H_N=-\frac{N}{2}(\mathbf{m}_N,\Delta \mathbf{m}_N )-N(\tilde{\mathbf{h}},\mathbf{m}_N )
\end{align}
where $(\cdot, \cdot)$ denote scalar product in $\mathbb{R}^K$ and 
\begin{align}\label{defn delta}
    \Delta= \boldsymbol\alpha \mathbf{J} \boldsymbol\alpha \qquad  \text{and } \qquad \tilde{\mathbf{h}}=\boldsymbol \alpha \mathbf{h}.
\end{align}
The joint distribution of $\s$ is governed by a Boltzmann-Gibbs measure
\begin{equation}\label{Gibbs}
     \mathcal{G}_{N} (\sigma) := \frac{e^{-  H_N(\sigma)}}{Z_N} \prod_{i=1}^{N} d\rho(\sigma_i),
\end{equation}
where $Z_N=\int_{[-1,1]^N} e^{-  H_N(\sigma)} \prod_{i=1}^{N} d\rho(\sigma_i)$ is the partition function. Averages w.r.t.\ $\mathcal{G}_N$ will be denoted with $\omega_{N}(\cdot)$. The thermodynamic pressure of the system is given by:
\begin{align}\label{thermo pressure GCW}
p_N = \frac{1}{N} \log Z_N .
\end{align}

\section{Main Results}
In this section, we provide the variational formula for the large $N$ limit of the thermodynamic pressure of the generalized multispecies Curie-Weiss model \eqref{MCW_hamiltonian_matrix}. In addition, in the case of Ising spins, we give central limit theorems for the magnetization vector $\mathbf{m}_N$. We stress again that the coupling matrix $\mathbf{J}$ used throughout the work is an arbitrary symmetric real matrix.

\subsection{Thermodynamic limit of the   pressure per particle}
Our first result is a representation of the large $N$ limit of, $p_N$, \eqref{thermo pressure GCW} in terms of a variational problem in $\mathbb{R}^K$. Let us define the following variational function: 
\begin{equation}\label{varfunctional}
p_{var}(\Delta, \mathbf{h}; \mathbf{x})\equiv p_{var}(\mathbf{x})= -\frac{1}{2}(\mathbf{x},\Delta\mathbf{x})+\sum_{p=1}^K\alpha_p\log \int_{\mathbb{R}} d\rho(\s) \exp \left[\s\left(\sum_{l=1}^KJ_{pl}\alpha_lx_l+h_p\right) \right] \, .
\end{equation}
Now, let $O$ be the orthogonal matrix diagonalizing $\Delta$.
Then the following theorem holds:

\begin{thm}\label{MCW_GENERIC} For any $(\mathbf{J},\mathbf{h})\in \mathbb{R}^{K\times K}\times\mathbb{R}^K $ one has 
\begin{align}\label{solution_MCW_GENERIC}
\lim_{N\to\infty}p_N=\inf_{z_1,\dots,z_a}\sup_{z_{a+1},\dots,z_K}
    p_{var}(O\mathbf{z}) \,.
\end{align}
\end{thm}

It is worth mentioning that the large $N$ limit of $p_N$ can be obtained through large deviations methods \cite{Eisele_Ellis_1988, Dembo_Zeitouni_2010, Opoku_Osabutey_2018}, while in this work we employ interpolation bounds.
We refer interested readers to \cite{Eisele_Ellis_1988}, where the phase diagram of the Curie-Weiss model with ferromagnetic interaction and generic compact spin distribution is analyzed.

%%%%%%%%%%%%%%%%%%%
%%%%%%%%%%
%%%%%%%%%%

\subsection{Fluctuations of the magnetization}
Here we state the fluctuations results of the species magnetization for the model with binary spins, i.e., $\sigma=(\s_i)_{i\leq N}\in \{-1,1\}^N$. Let us define 
\begin{equation}\label{def:functionf}
f(\mathbf{x}) = \frac{1}{2}(\mathbf{x},\Delta \mathbf{x}  )+(\tilde{\mathbf{h}},\mathbf{x} ) - ( \hat{\boldsymbol\alpha}, I(\mathbf{x}))\,,\quad \mathbf{x}\in[-1,1]^K
\end{equation}
where $I(\mathbf{x})=(I(x_p))_{p\leq K}\in\mathbb{R}^K$ with 
\begin{equation}\label{b entropy}
I(x)=\frac{1-x}{2}\log \left( \frac{1-x}{2}\right) +\frac{1+x}{2}\log\left(\frac{1+x}{2}\right) \,,\;\;x\in[-1,1] 
\end{equation}
and $\hat{\boldsymbol{\alpha}} $ is the vector associated to the diagonal matrix $\boldsymbol{\alpha}$. In the following, we assume  the convergence in \eqref{alphaconvergence} is fast enough. More precisely  setting $\boldsymbol\alpha_{N}=\text{diag}(\alpha_{N,p})_{p\leq K}$ we assume that 
\begin{equation}
\boldsymbol\alpha_{N} = \boldsymbol{\alpha}+ N^{-\theta}\text{diag}(\boldsymbol{\beta})
\end{equation}
where $\boldsymbol{\beta}=(\beta_p)_{p\leq K}$ for some $0<\beta_p<\infty$ and $\theta \in \left[\frac{1}{2}, \infty\right) $. The  fluctuations of the magnetization vector $\mathbf{m}_N$ depend on the global maximum point(s) of the function $f$ in \eqref{def:functionf}. A detailed analysis of the critical points and the maximizers  can be found  in \cite{Collet_2014, Gallo2008, Löwe_Schubert_2018} for $K=2$.  We mention that for general $K$ and  zero external fields, there is a  high temperature condition \cite{Knöpfel_Löwe_Schubert_Sinulis_2020, Kirsch_Toth_2022} that implies that the zero vector is the unique global maximum of $f$ and then the system has no spontaneous magnetization. The following statements contain central limit theorems for the vector of global species magnetization $\mathbf{m}_N=(m_1,...,m_K)$ with respect to the measure $\mathcal{G}_N$ \eqref{Gibbs}. Let us start with  the case of a unique non-degenerate global maximum. 

\begin{thm}\label{uniqness CLT}
Assume that  $(\mathbf{J},\mathbf{h})\in \mathbb{R}^{K\times K}\times\mathbb{R}^K $ are such that $f$ in \eqref{def:functionf} has a unique global maximizer  $\boldsymbol{\mu} = (\mu_r)_{r\leq K}$ with Hessian $\mathcal{H}_f(\boldsymbol{\mu})\prec 0$. Then we have the following convergence in law under the measure $\mathcal{G}_N$:
\begin{align}
%\boldsymbol{S}_N(\sigma) = 
    \left(\sqrt{N}\sqrt{\boldsymbol\alpha_{N}}(\mathbf{m}_N-\boldsymbol\mu)\right) \quad\xlongrightarrow[N \to\infty]{\mathcal{D}} \quad
    \begin{cases}
    \mathcal{N}(\boldsymbol\nu, \sqrt{\boldsymbol\alpha}\mathcal{H}^{-1}_{f}(\boldsymbol{\mu}) \sqrt{\boldsymbol\alpha})\; &\text{if }\theta = \frac{1}{2} \\
     \mathcal{N}(\mathbf{0}, \sqrt{\boldsymbol\alpha}\mathcal{H}^{-1}_{f}(\boldsymbol{\mu}) \sqrt{\boldsymbol\alpha})\; &\text{if }\theta > \frac{1}{2}
    \end{cases}
\end{align}
where $\boldsymbol\nu = -\sqrt{\boldsymbol\alpha}\; 
     \mathcal{H}_{f}^{-1}(\boldsymbol\mu) \,\boldsymbol\alpha \, \mathbf{J}\, \text{diag}(\boldsymbol{\mu})\boldsymbol{\beta}$.
\end{thm}
When instead multiple global maxima coexists, a proper condition is sufficient to restore Gaussianity:
\begin{thm}\label{nonuniqness CLT}
Assume that $(\mathbf{J},\mathbf{h})\in \mathbb{R}^{K\times K}\times\mathbb{R}^K $ are such that $f$ in \eqref{def:functionf} has $n$ global maximizers $\boldsymbol{\mu}^1, \ldots, \boldsymbol{\mu}^n$ with Hessian $\mathcal{H}_f(\boldsymbol{\mu}^i)\prec 0$ for all $i=1,\ldots, n$. For any collection $(A_i)_{i\leq n}$  subsets  of $[-1,1]^K$
such that $\boldsymbol\mu^{i}\in \text{int}(A_i)$ and $f(\boldsymbol\mu^i) > f(\mathbf{x})$ for all $\mathbf{x} \in \text{cl}(A_i)\backslash \{\boldsymbol\mu^i\}$ then

\begin{align}
%\boldsymbol{S}_N(\sigma) = 
    \left(\sqrt{N}\sqrt{\boldsymbol\alpha_{N}}(\mathbf{m}_N-\boldsymbol\mu^i)\right) |\{\mathbf{m}_N\in A_i\}\quad\xlongrightarrow[N\to\infty]{\mathcal{D}} \quad
    \begin{cases}
    \mathcal{N}(\boldsymbol\nu ^i, \sqrt{\boldsymbol\alpha}\mathcal{H}^{-1}_{f}(\boldsymbol{\mu}^i) \sqrt{\boldsymbol\alpha})\; &\text{if }\theta = \frac{1}{2}\\
     \mathcal{N}(\mathbf{0}, \sqrt{\boldsymbol\alpha}\mathcal{H}^{-1}_{f}(\boldsymbol{\mu}^i) \sqrt{\boldsymbol\alpha})\; &\text{if }\theta > \frac{1}{2}
    \end{cases}
\end{align}
and $\boldsymbol\nu ^i = -\sqrt{\boldsymbol\alpha}\; 
     \mathcal{H}_{f}^{-1}(\boldsymbol\mu ^i) \,\boldsymbol\alpha \, \mathbf{J}\, \text{diag}(\boldsymbol{\mu}^i)\boldsymbol{\beta}\;$
for all $i=1,\ldots, n$.
\end{thm} 
%\end{comment}

\section{Proofs}
\subsection{Proof of Theorem \ref{MCW_GENERIC}}
Let us start by observating that any indefinite matrix can be decomposed into the sum of two semi-definite matrices:
\begin{align}
    \Delta=\Delta_++\Delta_-\quad\text{where }\Delta_+\geq 0,\,\Delta_-\leq 0 .
\end{align}
Thus, the scalar products seen so far split into two parts of definite signs, for example:
\begin{align}
    (\mathbf{m}_N,\Delta\mathbf{m}_N) =
    (\mathbf{m}_N,\Delta_+\mathbf{m}_N)+
    (\mathbf{m}_N,\Delta_-\mathbf{m}_N) .
\end{align}
Also, we assume to order the species and eigenvectors so that the diagonal parts are as follows:
\begin{align}
&\Delta^D=\text{diag}(\lambda_1,\dots,\lambda_a,\lambda_{a+1},\dots,\lambda_K)\quad\text{with }\lambda_1,\dots,\lambda_a\leq 0,\,\lambda_{a+1},\dots,\lambda_K\geq 0\\\label{eq:eigenvalue_ordering}
    &\Delta^D=\Delta_+^D+\Delta_-^D=\text{diag}(\lambda_1,\dots,\lambda_a,0,\dots,0)+\text{diag}(0,\dots,0,\lambda_{a+1},\dots,\lambda_K) .
\end{align}
Here, the superscript $D$ indicates a diagonal matrix. Observe also that the orthogonal matrix $O$ diagonalizes simultaneously $\Delta,\,\Delta_+,\,\text{and} \, \Delta_-$. We are now in position to prove the following Sum Rule:

\begin{prop}[Hybrid Sum Rule]\label{Hybrid Sum Rule}
For every $N$, and $\mathbf{x}\in\mathbb{R}^K$ the following holds true:
\begin{multline}\label{HSR}
p_N=-\frac{1}{2}(\mathbf{x},\Delta_-\mathbf{x})+\frac{1}{N}\log\int_{\mathbb{R}^K}  d\rho(\s) e^{\frac{\beta N}{2}(\mathbf{m}_N,\Delta_+\mathbf{m}_N)+\beta N ( \mathbf{\tilde h}+\Delta_-\mathbf{x},\mathbf{m}_N)} +\\
+\frac{1}{2}\int_0^1dt\,\omega_{N,t}\Big[(
 \mathbf{m}_N-\mathbf{x},\Delta_-(\mathbf{m}_N-\mathbf{x}) )\Big] \, .
\end{multline}
where $\omega_{N,t}$ is the Gibbs state induced by a suitable interpolating Hamiltonian.
\end{prop}
%%%%%%%%%%%%%%%%%%%%
\begin{proof}
Let's introduce the following interpolating Hamiltonian:
\begin{align}\label{inter Hamiltonian}
    H_N(t)=-\frac{N}{2}(\mathbf{m}_N,\Delta_+ \mathbf{m}_N ) -\frac{Nt}{2}(\mathbf{m}_N,\Delta_- \mathbf{m}_N )-(1-t)(
    \Delta_-\mathbf{x},\mathbf{m}_N
    )-N (\mathbf{\tilde h},\mathbf{m}_N ).
\end{align}
The corresponding Boltzmann-Gibbs average is precisely $\omega_{N,t}$, and the interpolating pressure $p_N(t)$ has the following properties:
\begin{align}
    &p_N(1)=p_N \label{hybrid pn1}\\
    &p_N(0)=\frac{1}{N}\log\int_{\mathbb{R}^K} d\rho(\s) e^{\frac{\beta N}{2}(\mathbf{m}_N,\Delta_+\mathbf{m}_N)+\beta N (\mathbf{\tilde h}+\Delta_-\mathbf{x},\mathbf{m}_N)} \label{hybrid pn0} \,.
\end{align}
Completing the square easily proves that:
\begin{align*}
    p'_N(t)=-\frac{1}{N}\omega_{N,t}(H_N(t))=-\frac{1}{2}(\mathbf{x},\Delta_-\mathbf{x})+\frac{1
}{2}\omega_{N,t}\Big[(
    \mathbf{m}_N-\mathbf{x},\Delta_-(\mathbf{m}_N-\mathbf{x}))\Big] \,.
\end{align*}
The result follows from an application of the fundamental theorem of calculus.
\end{proof}

%\paragraph{\textbf{Proof of Corollary  \ref{MCW_upper_bound_nondisordered}}}
\begin{lemma}\label{MCW_upper_bound_nondisordered} 
Assume the eigenvalue ordering in \eqref{eq:eigenvalue_ordering}. Then for any $(\mathbf{J},\boldsymbol{h})\in \mathbb{R}^{K\times K}\times\mathbb{R}^K $ one has 
\begin{align}\label{upper_bound_nondisordered}
    p_N(\Delta,\mathbf{h})\leq \mathcal{O}\left(\frac{\log N}{N}\right) +
    \inf_{z_1\dots z_a}\sup_{z_{a+1}\dots z_K}p_{var}(\Delta,\mathbf{h};O\mathbf{z})\,.
    %+\sum_{p,l=1}^K \frac{2\Delta_{pl}}{N_pN_l} .
\end{align}
\end{lemma}
\begin{proof}
We define the grid of hypercubes $A^p_k=  \bigg[y^p_{k},y^p_{k}+\frac{2}{N_p}\bigg]$ where $y^p_k = -1+2(k-1)/N$ for $k\in\Lambda$, and $p$ is the species label. The vertices of this grid are identified by a multi-index $\gamma=(\gamma_1,\dots,\gamma_K)\in\Lambda^{K}$, and denoted $\mathbf{y}_\gamma$ in the following.

Furthermore, observe that:
\begin{multline}
    \frac{1}{N}\log\int_{\mathbb{R}^N}  d\rho(\s) e^{\frac{N}{2}(\mathbf{m}_N,\Delta_+\mathbf{m}_N)+N (\mathbf{\tilde h}+\Delta_-\mathbf{x},\mathbf{m}_N)}=\\
    =\frac{1}{N}\log\int_{\mathbb{R}^N}  d\rho(\s)\prod_{p=1}^K\sum_{k=1}^{N_p}\mathbbm{1}(m_{p}\in A_k^p)
    e^{\frac{N}{2}(\mathbf{m}_N,\Delta_+\mathbf{m}_N)+N (\mathbf{\tilde h}+\Delta_-\mathbf{x},\mathbf{m}_N)} \\
    =\frac{1}{N}\log \sum_{\gamma_1,\dots,\gamma_K=1}^{N_1,\dots,N_K}\int_{\mathbb{R}^{N}}d\rho(\sigma)\prod_{p=1}^K \mathbbm{1}(m_{p}\in A_{\gamma_p}^p)e^{\frac{N}{2}(\mathbf{m}_N,\Delta_+\mathbf{m}_N)+N (\mathbf{\tilde h}+\Delta_-\mathbf{x},\mathbf{m}_N)} \, .
\end{multline}
While the constraint $\prod_{p=1}^K \mathbbm{1}(m_{p}\in A_{\gamma_p}^p)$ is enforced, we have
\begin{align}
    0\leq \left(\mathbf{m}_N-\mathbf{y}_\gamma,\Delta_+(\mathbf{m}_N-\mathbf{y}_\gamma)\right) \leq\frac{C}{N^2}
\end{align}where $C>0$, and the multi-index $\gamma=(\gamma_1,\dots,\gamma_K)$. Using the bound above, we can write
\begin{align}
    \frac{1}{N}\log&\int_{\mathbb{R}^N}  d\rho(\s) e^{\frac{N}{2}(\mathbf{m}_N,\Delta_+\mathbf{m}_N) +N (\mathbf{\tilde h}+\Delta_-\mathbf{x},\mathbf{m}_N) }\leq \frac{C}{N^2}+\nonumber\\
    &+\frac{1}{N}\log \sum_{\gamma_1,\dots,\gamma_K=1}^{N_1,\dots,N_K}\int_{\mathbb{R}^{N}}d\rho(\sigma)\prod_{p=1}^K \mathbbm{1}(m_p \in A_{\gamma_p}^p)
    e^{-\frac{N}{2}(\mathbf{y}_\gamma,\Delta_+\mathbf{y}_\gamma) +N(\mathbf{m}_N,\Delta_+\mathbf{y}_\gamma) +N (\mathbf{\tilde h}+\Delta_-\mathbf{x},\mathbf{m}_N) }\nonumber\\
    &\leq \frac{C}{N^2}+\frac{1}{N}\log \sum_{\gamma_1,\dots,\gamma_K=1}^{N_1,\dots,N_K}\int_{\mathbb{R}^{N}}d\rho(\sigma)
    e^{-\frac{N}{2}(\mathbf{y}_\gamma,\Delta_+\mathbf{y}_\gamma) +N(\mathbf{m}_N,\Delta_+\mathbf{y}_\gamma) +N (\mathbf{\tilde h}+\Delta_-\mathbf{x},\mathbf{m}_N) }\nonumber\\
    %&\leq\mathcal{O}\left(\frac{\log N}{N}\right)+\frac{1}{N}\log \sup_{\mathbf{y}}\int_{\mathbb{R}^{N}}d\rho(\sigma) e^{-\frac{N}{2}(\mathbf{y},\Delta_+\mathbf{y}) +N(\mathbf{m}_N,\Delta_+\mathbf{y}) +N (\mathbf{\tilde h}+\Delta_-\mathbf{x},\mathbf{m}_N) } \nonumber \\
    & = \mathcal{O}\left(\frac{\log N}{N}\right)+\frac{1}{N}\log \sup_{\mathbf{y}}\int_{\mathbb{R}^{N}}d\rho(\sigma)
    e^{-\frac{N}{2}(\mathbf{y},\Delta_+\mathbf{y}) +N (\mathbf{\tilde h}+\Delta_-\mathbf{x} + \Delta_+\mathbf{y},\mathbf{m}_N) }\nonumber \\
    & = \mathcal{O}\left(\frac{\log N}{N}\right)+\frac{1}{N}\log \sup_{\mathbf{y}} e^{-\frac{N}{2}(\mathbf{y},\Delta_+\mathbf{y}) }
    \int_{\mathbb{R}^{N}}d\rho(\sigma) 
    e^{N \sum_{p=1}^K \alpha_p \frac{1}{N_p}\sum_{i\in \Lambda_p}\sigma_i ( h_p + \boldsymbol{\hat{\alpha}}^{-1} (\Delta_{-}\mathbf{x} + \Delta_+\mathbf{y}) _p) }\nonumber\\
      = \sup_{\mathbf{y}}&\left\{-\frac{(\mathbf{y},\Delta_+\mathbf{y}) }{2}+
    \sum_{p=1}^K \alpha_p \log \int_{\mathbb{R}} d\rho(\s) \exp\{\s[\mathbf{h} + \boldsymbol{\hat{\alpha}}^{-1}(\Delta_-\mathbf{x} +\Delta_+ \mathbf{y} )  ]_p \} \right\} + \mathcal{O}\left(\frac{\log N}{N} \right) \,.
\end{align}

Now, thanks to the decomposition of $\Delta$, there are some combinations of components of $\mathbf{x}$ and $\mathbf{y}$ that will never appear, i.e., those corresponding to the zero eigenvalues of $\Delta_+\, \text{and} \, \Delta_-$. Therefore, we have that:
\begin{align}
    &\Delta_+\mathbf{y}+\Delta_-\mathbf{x}=O\Delta^D_+\mathbf{y}^D+O\Delta^D_-\mathbf{x}^D=O\Delta^D \mathbf{z}^D=\Delta \mathbf{z}\\
    & \text{and} \quad \mathbf{z}^D=(x_1^D,\dots,x_a^D,y_{a+1}^D,\dots,y_K^D)
\end{align}
where $D$, when accompanying vectors, indicates vectors read in the basis of the eigenvectors of $\Delta$. Hence, the real degrees of freedom are only $K$. Now, by equation \eqref{hybrid pn0}, we have that: 
\begin{multline}
    p_N(0)\leq \mathcal{O}\left(\frac{\log N}{N}\right) +\sup_{y_{a+1}^D,\dots,y_K^D}p_{var}\left(\Delta_+,\boldsymbol{\hat\alpha}^{-1}\Delta_-\mathbf{x}+\mathbf{h};O\mathbf{y}^D\right)\\
    =\mathcal{O}\left(\frac{\log N}{N}\right) +\sup_{y_{a+1}^D,\dots,y_K^D}\left\{
    -\frac{1}{2}\left( O\mathbf{y}^D,\Delta_+O\mathbf{y}^D\right)
    +\sum_{p=1}^K\alpha_p\log\int_{\mathbb{R}} d\rho(\s) e^{\s (\boldsymbol{\hat\alpha}^{-1}\Delta O\mathbf{z}^D
    +\mathbf{h})_p }
    \right\} \, .
\end{multline}

Notice that this holds for any $x^D_1,\dots,x^D_a$ corresponding to the negative definite matrix. If we add the missing quadratic term from the derivative, which contains only $\Delta_-$, and optimize over the remaining degrees of freedom, we get:
\begin{align*}
    p_N \leq \mathcal{O}\left( \frac{\log N}{N}\right) +\inf_{z_1,\dots,z_a}\sup_{z_{a+1},\dots,z_K}
    p_{var}(O\mathbf{z}) \, .
\end{align*}
We removed $D$ because now $\mathbf{z}$ is a dummy variable.
\end{proof}

Now we need to find the other bound.\\
%%

%\paragraph{\textbf{Proof of Theorem \ref{MCW_GENERIC}}}
\begin{lemma}\label{MCW lowerbound} For any $(\mathbf{J},\boldsymbol{h})\in \mathbb{R}^{K\times K}\times\mathbb{R}^K $ one has 
\begin{align}\label{solution_nondisordered_MCW_GENERIC}
p_N \geq \inf_{z_1,\dots,z_a}\sup_{z_{a+1},\dots,z_K}
    p_{var}(O\mathbf{z}) + \mathcal{O}\left( \frac{1}{N} \right) \, .  %\frac{1}{2}\int_0^1dt \;\omega'_{N,t}\Big[(\mathbf{m}_N-\mathbf{z},\Delta(\mathbf{m}_N-\mathbf{z}))\Big].
\end{align}
\end{lemma}
\begin{proof}
For any $\mathbf{z}\in\mathbb{R}^K$ consider the  interpolating Hamiltonian:
\begin{align}\label{inter Hamiltonian g}
    H_N(t)= -\frac{N}{2}(\mathbf{m}_N,\Delta\mathbf{m}_N )
    +\frac{Nt}{2} (\mathbf{m}_N-\mathbf{z},\Delta(\mathbf{m}_N-\mathbf{z}))-N(\mathbf{\tilde h},\mathbf{m}_N )\,.
\end{align}
The corresponding pressure $p_N(t)$ has the following properties:
\begin{align}
    &p_N(0)=p_N \label{hyb pn1}\\
    &p_N(1 )= p_{var}(\Delta,\mathbf{h};\mathbf{z})\\
   & p_N'(t)=-\frac{1}{2}\omega_{N,t}\left[(\mathbf{m}_N-\mathbf{z},\Delta(\mathbf{m}_N-\mathbf{z}))\right] \, .
\end{align}
Using the convexity of the interpolating pressure $p_N(t)$, one has $p_N(1)\leq p_N(0)+p'_N(0)$ and then  for every $\mathbf{z}\in\mathbb{R}^K$, we obtain:
\begin{align}
    p_N \geq
    \frac{1}{2}\omega_{N,0}[(
    \mathbf{m}_N-\mathbf{z},\Delta(\mathbf{m}_N-\mathbf{z})
    )]+p_{var}(\Delta,\mathbf{h};\mathbf{z}).
\end{align}
Now suppose that $\mathbf{z}$ equal to a  critical point $\mathbf{\bar{z}}^D=O^{-1}\mathbf{\bar{z}}\;$ of $\;p_{var}$, then:  
\begin{multline}\label{critical}
    \frac{\partial}{\partial \mathbf{z}^D}p_{var}(O\mathbf{\bar{z}}^D)
    =O\frac{\partial}{\partial\mathbf{z}}p_{var}(\mathbf{\bar{z}})
    =O\left[-\Delta\mathbf{\bar{z}}+ \frac{\Delta\int_{\mathbb{R}} \sigma \exp \s( \boldsymbol{\alpha}^{-1}\Delta\mathbf{\bar{z}}+\mathbf{h}) d\rho(\s)}{\int_{\mathbb{R}}\exp \tau ( \boldsymbol{\alpha}^{-1}\Delta\mathbf{\bar{z}}+\mathbf{h}) d\rho(\tau)} \right]=0 
\end{multline}where the above exponentials are applied component-wise, and division is also component-wise. Criticality then implies from \eqref{critical} that
\begin{align}
\mathbf{\bar{z}}- \frac{\int_{\mathbb{R}} \sigma \exp \s( \boldsymbol{\alpha}^{-1}\Delta\mathbf{\bar{z}}+\mathbf{h}) d\rho(\s)}{\int_{\mathbb{R}}\exp \tau ( \boldsymbol{\alpha}^{-1}\Delta\mathbf{\bar{z}}+\mathbf{h}) d\rho(\tau)}\in\text{Ker}(\Delta).
\end{align}
Observe that the above can also be recast as
\begin{align*}
    &\mathbf{\bar{z}}-\omega_{N,0}(\mathbf{m}_N)\in\text{Ker}(\Delta).
\end{align*}
Let us now expand the quadratic form and separate the diagonal terms from the off-diagonal ones:
\begin{multline}
\omega_{N,0}[(
    \mathbf{m}_N-\mathbf{\bar{z}},\Delta(\mathbf{m}_N-\mathbf{\bar{z}})
    )]=\sum_{p\neq l,1}^K \Delta_{pl} \omega_{N,0}[(m_p -\bar{z}_p)] \omega_{N,0}[(m_{N_l}-\bar{z}_l)]+\\
    +\sum_{p=1}^K \Delta_{pp} \omega_{N,0}[(m_p -\bar{z}_p)^2] \,.
\end{multline}

Now, notice that
\begin{multline}
\omega_{N,0}[(m_p -\bar{z}_p)^2] = \omega_{N,0}[(m_p )^2]-2\omega_{N,0}[m_p ]\bar{z}_p+(\bar{z}_p)^2\\
=\frac{1}{N^2_p}\left(\sum_{i\in \Lambda_p}\omega_{N,0}[\sigma^2_i]+\sum_{i\neq j} \omega_{N,0}[\sigma_i\sigma_j]\right)-2\omega_{N,0}[m_p ]\bar{z}_p+(\bar{z}_p)^2 \, .
\end{multline}
For $i\neq j$ the measure $\omega_{N,0}$ factorizes, and naturally $\sum_{i\in\Lambda_p}\sigma_i^2 \leq N_p$. Therefore
\begin{equation}
\omega_{N,0}[(m_p -\bar{z}_p)^2] = \omega^2_{N,0}[(m_p -\bar{z}_p)] +\mathcal{O}(N^{-1}).
\end{equation}
The previous finally implies
 \begin{equation}
    \omega_{N,0}\Big[(
    \mathbf{m}_N-\mathbf{\bar{z}},\Delta(\mathbf{m}_N-\mathbf{\bar{z}})
    )\Big]=\mathcal{O}\left( \frac{1}{N}\right) +
    (\omega_{N,0}(\mathbf{m}_N -\mathbf{\bar{z}}),\Delta\omega_{N,0}(\mathbf{m}_N -\mathbf{\bar{z}}))\,.
\end{equation}
The last term is zero because $\omega_{N,0}(\mathbf{m}_N -\mathbf{\bar{z}})$ lies in the kernel of $\Delta$.  
\end{proof}

The proof of Therorem \ref{MCW_GENERIC} thus follows from Lemma \ref{MCW_upper_bound_nondisordered} and Lemma \ref{MCW lowerbound}.

\subsection{Proof of Theorem \ref{uniqness CLT}}
The proof of the CLT in the case of Ising spins, namely $\rho=\frac{1}{2}(\delta_{-1}+\delta_{-1})$  is based on a careful control of the asymptotic expansion of a  partition function $Z_N$ using the methods of \cite{Contucci_Mingione_Osabutey_2024, MSB21}. For any integer $N$ and $x\in[-1,1]$ we define the quantity

\begin{equation}
A_{N}(x) =\mathrm{card}\Big\{\s^{}\in \{-1,1\}^{N}: \frac{1}{N}\sum_{i=1}^N\sigma_i =x\Big\} = \binom{N}{ \frac{N(1+x)}{2}} .
\end{equation}
Let's state some useful bounds on $A_N$, the proof is standard and can be found in \cite{Talagrand2003spin}.

\begin{lemma}\label{configuration count}
For any $x\in[-1,1]$ the following inequality holds:
\begin{equation}\label{boundAN}
\dfrac{1}{C\sqrt{N}}e^{-NI(x)}\leq A_N(x) \leq e^{-NI(x)}
\end{equation}
where $C$ is a universal constant and, 
\begin{equation}\label{binaryentropy}
I(x)=\frac{1-x}{2}\log\left( \frac{1-x}{2}\right) +\frac{1+x}{2}\log\left( \frac{1+x}{2}\right)   .
\end{equation}
Moreover for any $x\in(-1,1)$ one has 
\begin{equation}\label{stirlings}
A_N(x)= \sqrt{\frac{2}{\pi N(1-x^2)}} 
 \exp{\left( -N I(x) \right) } \cdot \left(1 + \mathcal{O}\left(N^{-1}\right)\right).%\cr
\end{equation}
\end{lemma}

Let us start by dividing the configuration space
$\{-1,1\}^N$ into microstates of equal local magnetization.  For a given species $l\leq K$ with spins  configuration $\s^{(l)}$, the local magnetization $m_l$ takes values in  $S_{l}=\{-1+\frac{2n}{N_l},n=0,\ldots, N_l\}$, with $|S_{l}|= (N_l+1)$. The possible values of the magnetization vector $\mathbf{m}_N=(m_{l})_{l\leq K} \in S_N=\bigtimes_{l=1}^K S_l$.  
Hence  the partition function rewrites as:
\begin{equation}\label{partition fnx}
Z_N = \sum_{\mathbf{x}\in S_{N}} \prod_{l=1}^K A_{N_l}(x_{l}) \exp{\left( - H_N(\mathbf{x}) \right) } \, .
\end{equation}
Here, $A_{N_l}(x_{l})$ counts the number of all possible configurations of $\s^{(l)} \in \{-1,1\}^{N_l}$ that share the same magnetization $x_l$. From Lemma \ref{configuration count} one can obtain the following bound for the pressure, by substituting \eqref{boundAN} into \eqref{partition fnx}:
\begin{equation}\label{bound pressure}
-\frac{1}{N}\left( \log{C} + \frac{1}{2}\sum_{l=1}^K\log N_l\right)  + \max_{\mathbf{x}} f_N(\mathbf{x}) \leq p_N \leq \frac{1}{N} \sum_{l=1}^K\ \log(N_l+1) + \max_{\mathbf{x}} f_N(\mathbf{x})
\end{equation}
where
\begin{equation}\label{f_N}
f_N(\mathbf{x}) 
 = \frac{1}{2}\sum_{p,l=1}^K\alpha_{N, l} J_{pl}\alpha_{N, p}x_px_l + \sum_{p=1}^K\alpha_{N, p}h_px_p - \sum_{p=1}^K\alpha_{N, p} I(x_p) .
\end{equation}Therefore
\begin{equation}\label{tlimitpre}
 \lim_{N\to \infty} p_N =\max_{\mathbf{x}\in[-1,1]^K} f(\mathbf{x})
 \end{equation}
where 
\begin{equation}\label{functionf}
f(\mathbf{x}) = \lim_{N\to \infty} f_N(\mathbf{x}) = \frac{1}{2}(\mathbf{x},\Delta \mathbf{x}  )+(\tilde{\mathbf{h}},\mathbf{x} ) - (\boldsymbol{\hat\alpha}, I(\mathbf{x}))
\end{equation}
and $I(\mathbf{x})=(I(x_l))_{l\leq K}$ is defined in \eqref{binaryentropy}. The  stationarity conditions are:
\begin{equation}\label{Kgroup MFE}
    x_l = \tanh{\left( h_l + \sum_{p=1}^K\alpha_pJ_{lp}x_p \right) } \quad \quad \text{for} \quad l=1,...,K. 
\end{equation}

The solutions of the fixed point equation \eqref{Kgroup MFE} identify the stationary points of $f$ among which we are interested in the ones that reach the supremum.  There can be more than one  global maximizer depending on the parameters  $(\mathbf{J}, \mathbf{h})$ of the model. Here we assume that $(\mathbf{J}, \mathbf{h})$ are such that there exists a unique maximizer  $\boldsymbol{\mu}=(\mu^{(l)})_{l\leq K}$ of $f$ with negative definite Hessian, i.e., $\mathcal{H}_{f}(\boldsymbol{\mu})\prec 0$. We also stress that $f$ is smooth around $\boldsymbol{\mu}$ since the latter belongs to  the interior of  $[-1,1]^K$ for any choice of $(\mathbf{J},\mathbf{h})$ (see  Lemma \eqref{maximizer mu}).

The main idea is to derive  the asymptotic of the partition function for a perturbed system, where a small external field is added. For any $\mathbf{t}\in\mathbb{R}^K$ we define:
\begin{equation}\label{H_nt}
H_{N,\mathbf{t}}(\mathbf{m}_N)= H_{N}(\mathbf{m}_N)- 
\sqrt{N}\left(\mathbf{t},\sqrt{\boldsymbol{\alpha}_N}\, \mathbf{m}_N\right)
\end{equation}
where $\sqrt{\boldsymbol{\alpha}_N} = \text{diag}(\sqrt{N_p/N})_{p\leq K}$, and $H_{N}(\mathbf{m}_N)$ is the unperturbed Hamiltonian defined in \eqref{MCW_hamiltonian_matrix} with $Z_{N,\mathbf{t}}$ as the associated partition function. We also define the function:
\begin{equation}\label{f_nt alpha}
f_{N,\mathbf{t}}(\mathbf{x})= f_{N}(\mathbf{x})+ \frac{1}{\sqrt{N}}\left(\mathbf{t}, \sqrt{\boldsymbol\alpha_N}\mathbf{x}\right) \, .
\end{equation}
Let us notice from equation \eqref{f_nt alpha} that  $f_{N,\mathbf{t}}$ and all its partial derivatives with respect to $\mathbf{x}$ at any order converge uniformly to the one of  $f$ for any $\mathbf{x}$ in the interior of $[-1,1]^K$. 
Therefore by Lemma \ref{uniform convergence}, for $N$ large enough   $f_{N,\mathbf{t}}$ has a unique maximizer $\boldsymbol{\mu}_{N,\mathbf{t}} = (\mu_{N,\mathbf{t}}^{(l)})_{l\leq K} \in \mathbb{R}^K$ that converge to $\boldsymbol{\mu}$ and $\mathcal{H}_{f_{N,\mathbf{t}}}(\boldsymbol{\mu}_{N,\mathbf{t}})\prec 0$.

Next we will show that the magnetization vector concentrates around $\boldsymbol{\mu}_{N,\mathbf{t}}$ with overwhelming probability with respect to the Boltzmann-Gibbs measure induced by $H_{N,\mathbf{t}}$, i.e., $\mathcal{G}_{N,\mathbf{t}}$. For $\delta>0$ we denote by $ B_{N,\delta}$ a poly-interval centered at $ \boldsymbol\mu_{N,\mathbf{t}}$ with each coordinate $\mu_{N,\mathbf{t}}^{(l)}$ as the center and $N^{-\frac{1}{2}+\delta}_l$ as the length of the interval along the $l${th} dimension, namely:
\begin{equation}\label{polyinterval}
    B_{N,\delta} = \left\{ \mathbf{x} \in \mathbb{R}^K : \; \left|x_l - \mu_{N,\mathbf{t}}^{(l)}\right| < N^{-\frac{1}{2}+\delta}_l, \;\forall \;l\in \{1,\ldots,K\} \right\} \, .
\end{equation}

%%%%%%%%%%%%%%%%%%%
%%%%%%%%%%%%%%%%%%%
\begin{lemma}\label{lemma conc uniq}
Assume that $f(\mathbf{x})$ has a unique global maximizer $\boldsymbol\mu$ with $\mathcal{H}_f(\boldsymbol\mu)\prec 0$. Then for $N$ large enough $f_{N,\mathbf{t}}$ has a unique maximizer $\boldsymbol\mu_{N,\mathbf{t}}\xrightarrow{N\to\infty} \boldsymbol\mu$ and $\mathcal{H}_{f_{{N,\mathbf{t}}}}(\boldsymbol\mu_{N,\mathbf{t}})\prec 0$. Moreover  for $\delta\in(0,\frac{1}{2K+4})$ we have that
\begin{equation}\label{conc1}
\mathcal{G}_{N,\mathbf{t}}(\mathbf{m}_N\in B_{N,\delta}^c({\boldsymbol\mu}_{N,\mathbf{t}})) \leq \exp\bigg\{\frac{1}{2}N^{2\delta}\lambda_{N,\mathbf{t}}\bigg\}\mathcal{O}\left( N^{\frac{3K}{2}}\right) .
\end{equation}
where $\lambda_{N,\mathbf{t}}<0$ and the partition function \eqref{partition fnx} can be expanded as:
\begin{equation}\label{part_fnx exp}
%\begin{split}
Z_{N,\mathbf{t}} = \frac{e^{N f_{N,\mathbf{t}}(\boldsymbol\mu_{N,\mathbf{t}})}}{\sqrt{\det{\left(-\mathcal{H}_{{f}_{N,\mathbf{t}}} (\boldsymbol\mu_{N,\mathbf{t}})\right)} \prod_{l=1}^K(1-(\mu_{N,\mathbf{t}}^{(l)})^2)}} \cdot \bigg( 1+O\left(N^{-\frac{1}{2}+(K+2)\delta}\right) \bigg).
%\end{split}
\end{equation}
\end{lemma}

\begin{proof}

The fact that for $N$ large enough 
$f_{N,\mathbf{t}}$ has a unique non degenrate global maximizers follows easily by Lemma \ref{uniform convergence}. For the second part of the statement one has 
\begin{equation}\label{prob_m*}
\begin{split}
    \mathcal{G}_{N,\mathbf{t}}(\mathbf{m}_N\in B_{N,\delta}^c) &= \dfrac{\sum_{\mathbf{x}\in S_N\cap B_{N,\delta}^c}\prod_{l=1}^K A_{N_l}(x_{l}) \exp{\left( - H_{N,\mathbf{t}}(\mathbf{x}) \right) }}{\sum_{\mathbf{x}\in S_N}\prod_{l=1}^K A_{N_l}(x_{l}) \exp{\left( - H_{N,\mathbf{t}}(\mathbf{x}) \right) }}\\
    & \leq \frac{C \prod_{l=1}^K \sqrt{N_l}(N_l+1)\sup_{\mathbf{x}\in B_{N,\delta}^c}e^{Nf_{N,\mathbf{t}}(\mathbf{x})}}{\sup_{\mathbf{x}\in [-1,1]^K}e^{Nf_{N,\mathbf{t}}(\mathbf{x})}}\\
    & = \exp{\left\{N\left(\sup_{\mathbf{x}\in B_{N,\delta}^c}f_{N,\mathbf{t}}(\mathbf{x}) - f_{N,\mathbf{t}}({\boldsymbol\mu}_{N,\mathbf{t}}) \right)\right\}} \mathcal{O}\left(  N^{\frac{3K}{2}}\right) .
\end{split}
\end{equation}

The $\sup$ in the denominator of the inequality appearing in the second line has been extended over $[-1,1]^K$ instead of $S_N$ due to the fact that $f_{N,\mathbf{t}}$ is Lipschitz continuous away from the boundaries of $[-1,1]^K$. The difference between  $\sup_{\mathbf{x}\in S_N}f_{N,\mathbf{t}}(\mathbf{x})$ and $\sup_{\mathbf{x}\in [-1,1]^K}f_{N,\mathbf{t}}(\mathbf{x})$ is small due to the Lipschitz property of $f_{N,\mathbf{t}}$ and bounded by a constant due to the mesh size of $S_N$. Let us consider the simplest case where $K=1$. Then the set $S_N = S_1$ and $S_{1}=\left\{-1+\frac{2n}{N_1},n=0,\ldots, N_1\right\}$ with mesh size $\frac{2}{N_1}$ and $\mathbf{t} \in \mathbb{R}^K = t$. Now, the Lipschitz continuity of $f_{N,t}$ means there exists a constant $L$ such that: $\left| f_{N,t}(x) - f_{N,t}(y)\right|\leq L|x-y|$ for all $x,y$ away from $\pm 1$. Given that the maximum distance between any point $x\in S_N$ and the nearest point in $(-1,1)$ is at most $\frac{2}{N_1}$, the difference between the supremum of $f_{N,t}$ over $S_N$ and over $(-1,1)$ can be bounded by: $\left| \sup_{x\in S_N} f_{N,t}(x) - \sup_{x\in (-1,1)} f_{N,t}(x)\right|\leq \frac{2L}{N_1}$. Thus, the difference is small and goes to zero as $N_1 \to \infty$. This justifies substituting the sum over the discrete set $S_N$ with the supremum over $[-1,1]$, with the error being controlled by the mesh size and the Lipschitz constant $L$. 

We claim that in the last equality of \eqref{prob_m*}, the supremum  of $f_{N,\mathbf{t}}$ restricted to $B_{N,\delta}^c$ is attained at some $\mathbf{y}_{N,\mathbf{t},\delta}$, lying on  the boundary of $B_{N,\delta}$. A formal  proof can be obtained through a straighforward generalization of  Lemma B.11 in \cite{mukherjee2021variational}. In fact, if  $\mathbf{x}\in B_{N,\delta}^c$ then at least one of the coordinates $x_l$ is outside $\left(\mu_{N,\mathbf{t}}^{(l)}-N_l^{-\frac{1}{2}+\delta}, \mu_{N,\mathbf{t}}^{(l)}+N
_l^{-\frac{1}{2}+\delta}\right)$. The function $f_{N,\mathbf{t}}$  decreases outside $B_{N,\delta}$, and hence the supremum must be on the boundary. 
Bearing this in mind we have: 
\begin{equation}
|\mathbf{y}_{N,\mathbf{t},\delta}- \boldsymbol\mu_{N,\mathbf{t}}|^2\leq \sum_{l\leq K} N_l^{-1+2\delta}=\sum_{l\leq K} (N\alpha_{N, l})^{-1+2\delta }= N^{-1+2\delta } \sum_{l\leq K} (\alpha_{N, l})^{-1+2\delta }=
N^{-1+2\delta } c_{\alpha_N,\delta}
\end{equation}
and, hence, using a Taylor expansion up to third order with Lagrange type remainder:
\begin{multline}\label{taylor conc1}
f_{N,\mathbf{t}}\left( \mathbf{y}_{N,\mathbf{t},\delta} \right)  =  f_{N,\mathbf{t}}(\boldsymbol{\mu}_{N,\mathbf{t}})+\frac{1}{2}\left( (\mathbf{y}_{N,\mathbf{t},\delta}- \boldsymbol\mu_{N,\mathbf{t}})\mathcal{H}_{f_{N,\mathbf{t}}}(\boldsymbol\mu_{N,\mathbf{t}}),(\mathbf{y}_{N,\mathbf{t},\delta}- \boldsymbol\mu_{N,\mathbf{t}})\right) + \\
+ \frac{1}{6} \sum_{l,p,s=1}^K \frac{\partial^3 f_{N,\mathbf{t}}}{\partial x_l \partial x_p \partial x_s}(\vartheta_{N,\mathbf{t}}) (\mathbf{y}_{N,\mathbf{t},\delta} - \boldsymbol{\mu}_{N,\mathbf{t}})_l (\mathbf{y}_{N,\mathbf{t},\delta} - \boldsymbol{\mu}_{N,\mathbf{t}})_p (\mathbf{y}_{N,\mathbf{t},\delta} - \boldsymbol{\mu}_{N,\mathbf{t}})_s\\
\leq \,f_{N,\mathbf{t}}(\boldsymbol{\mu}_{N,\mathbf{t}}) + \frac{1}{2} N^{-1+2\delta} \underbrace{c_{\alpha_N,\delta}}_{>1}\lambda_{N,\mathbf{t}} + \frac{1}{6} \sum_{l,p,s=1}^K \left|\frac{\partial^3 f_{N,\mathbf{t}}}{\partial x_l \partial x_p \partial x_s}(\vartheta_{N,\mathbf{t}})\right| \cdot N^{-3/2 + 3\delta} (c_{\alpha_N,\delta})^{3/2} \\
\leq \,f_{N,\mathbf{t}}(\boldsymbol{\mu}_{N,\mathbf{t}}) + \frac{1}{2} N^{-1+2\delta}\lambda_{N,\mathbf{t}} + \mathcal{O}\left( N^{-3/2 + 3\delta}\right) 
\end{multline}
where $\vartheta_{N,\mathbf{t}}$ is an intermediate point of the segment $[\boldsymbol{\mu}_{N,\mathbf{t}},\mathbf{y}_{N,\mathbf{t},\delta}]$, $\lambda_{N,\mathbf{t}}<0$ is the largest eigenvalue of $\mathcal{H}_{f_{N,\mathbf{t}}}$ and in the third term of the last inequality of \eqref{taylor conc1} we have used that $\alpha_{N, l} \to \alpha_l$ as $N\to \infty$. Hence, \eqref{conc1} follows from \eqref{prob_m*} and \eqref{taylor conc1}.

Let's now begin with the proof of the asymptotic expansion of the partition function $Z_{N,\mathbf{t}}$ \eqref{part_fnx exp}. Observe that the concentration results, equation \eqref{conc1}, implies that almost all the contribution to $Z_{N,\mathbf{t}}$ comes from spin configurations having magnetization in a vanishing neighbourhood of the maximizer $\boldsymbol\mu_{N,\mathbf{t}}$, i.e., 
$\mathcal{G}_{N,\mathbf{t}}(\mathbf{m}_N\in B_{N,\delta})= 1-\mathcal{O}(e^{-cN^{2\delta}})$ for some $c>0$. Hence,
\begin{equation}\label{partition fnx exp}
Z_{N,\mathbf{t}} =\left(1+\mathcal{O}\left(e^{-cN^{2\delta}}\right)\right) \sum_{\mathbf{x}\in S_N\cap B_{N,\delta}} \zeta_{N,\mathbf{t}}(\mathbf{x})\,,\quad
\zeta_{N,\mathbf{t}}(\mathbf{x}) := \prod_{l=1}^K 
\binom{N_l}{\frac{N_l(1+x_l)}{2}} \exp\left( -H_{N,\mathbf{t}}(\mathbf{x})\right) .
\end{equation}

Following the same argument of \cite{Contucci_Mingione_Osabutey_2024,MSB21}, we approximate the sum in \eqref{partition fnx exp} by an integral using Lemma \ref{appendix riemann} over the set $B_{N,\delta}$ with shrinking interval containing the unique vector of global maximizer of $f_{N,\mathbf{t}}$ which are elements of $S_N$:
\begin{multline}\label{diff zeta}
%\begin{split}
    \Bigg|\int_{B_{N,\delta}} \zeta_{N,\mathbf{t}}(\mathbf{x}) d\mathbf{x} - \frac{2^K}{\prod_{l=1}^K N_l}\sum_{\mathbf{x}\in S_{N}\cap B_{N,\delta} } \zeta_{N,\mathbf{t}}(\mathbf{x}) \Bigg| 
    \leq N^{\left( {\frac{1}{2}+\delta}\right)  \left( K-1\right) }\cdot N^{-\frac{1}{2}+\delta}\cdot N^{-K}\sup_{\mathbf{x}\in B_{N,\delta}}|\nabla\zeta_{N,t}(\mathbf{x})|\\
= \mathcal{O}\left( N^{\left( {-\frac{1}{2}+\delta}\right)  \left( K+1\right) }\right) \zeta_{N,\mathbf{t}}(\boldsymbol\mu_{N,\mathbf{t}})
%\end{split}
\end{multline}
where $\sup_{\mathbf{x}\in B_{N,\delta}}|\nabla\zeta_{N,\mathbf{t}}(\mathbf{x})|$ is bounded in Lemma \ref{sup derivative zeta'}. Now, following from \eqref{diff zeta} and applying the results of Lemma \ref{appendix laplace} to approximate the integral, we have that:

\begin{equation}
\begin{split}
\sum_{\mathbf{x}\in S_{N}\cap B_{N,\delta} } & \zeta_{N,\mathbf{t}}(\mathbf{x}) = \frac{\prod_{l=1}^K N_l}{2^K} \int_{B_{N,\delta}} \zeta_{N,\mathbf{t}}(\mathbf{x}) d\mathbf{x} + \mathcal{O}\left( N^{\left( {-\frac{1}{2}+\delta}\right)  \left( K+1\right)  +K}\right)  \zeta_{N,\mathbf{t}}(\boldsymbol\mu_{N,\mathbf{t}})\cr 
%=&\Big[1+\mathcal{O}\left( N^{-1}\right) \Big]\;\frac{\prod_{l=1}^K N_l}{2^K} \int_{B_{N,\delta}} \sqrt{\frac{2^K}{\pi^K\prod_{l=1}^KN_l(1-x_l^2)}}\; e^{N f_{N,\mathbf{t}}(\mathbf{x})} \; d\mathbf{x} \cr
%+& \mathcal{O}\left( N^{{\frac{1}{2}\left( K-1\right) }+\delta\left( 1+K\right) }\right) \Big[1+\mathcal{O}\left( N^{-1}\right) \Big] \sqrt{\frac{2^K}{\pi^K\prod_{l=1}^KN_l(1-\mu_{N,\mathbf{t}}^{(l)}^2)}}\; e^{N f_{N,\mathbf{t}}(\boldsymbol\mu_{N,\mathbf{t}})}\cr
=& \frac{\prod_{l=1}^K N_l}{2^K} \sqrt{\frac{2^K}{\pi^K\prod_{l=1}^KN_l(1-(\mu_{N,\mathbf{t}}^{(l)})^2)}}\sqrt{\frac{(2\pi)^K}{\prod_{l=1}^KN_l\det{\left(-\mathcal{H}_{f_{N,\mathbf{t}}} (\boldsymbol\mu_{N,\mathbf{t}})\right)} }} e^{N f_{N,\mathbf{t}}(\boldsymbol\mu_{N,\mathbf{t}})}\cr
&\cdot \bigg( 1+\mathcal{O}\left(N^{-1/2+(K+2)\delta}\right) \bigg) \cr
&= \frac{e^{N f_{N,\mathbf{t}}(\boldsymbol\mu_{N,\mathbf{t}})}}{\sqrt{\det{\left(-\mathcal{H}_{{f}_{N,\mathbf{t}}} (\boldsymbol\mu_{N,\mathbf{t}})\right)} \prod_{l=1}^K(1-(\mu_{N,\mathbf{t}}^{(l)})^2)}} \cdot \bigg( 1+\mathcal{O}\left(N^{-\frac{1}{2}+(K+2)\delta}\right) \bigg)
\end{split}
\end{equation}
for  $\delta\in(0,\frac{1}{2K+4})$. Hence, the partition function becomes:
\begin{equation}\label{part_fnx asymp exp}
%\begin{split}
Z_{N,\mathbf{t}} = \frac{e^{N f_{N,\mathbf{t}}(\boldsymbol\mu_{N,\mathbf{t}})}}{\sqrt{\det{\left(-\mathcal{H}_{{f}_{N,\mathbf{t}}} (\boldsymbol\mu_{N,\mathbf{t}})\right)} \prod_{l=1}^K(1-(\mu_{N,\mathbf{t}}^{(l)})^2)}} \cdot \bigg( 1+\mathcal{O}\left(N^{-\frac{1}{2}+(K+2)\delta}\right) \bigg).
%\end{split}
\end{equation}
This completes the proof.
\end{proof}

Now we are ready to proof Theorem \ref{uniqness CLT}. In order to approximate the distribution of the scaled difference between the vector of global species magnetization and the limiting global maximizers, we will compute the  limiting  moment generating function of $\mathbf{m}_N=(m_{l})_{l\leq K}$ for some $\mathbf{t}\in \mathbb{R}^K$ using the expanded form of the partition function in \eqref{part_fnx asymp exp}: 
%%%%%%%%%%%%%%%%%%%%%%%%%%%%
\begin{multline}\label{mgf}
\mathbb{E}\bigg[e^{\sqrt{N}( \mathbf{t},\sqrt{\boldsymbol{\alpha}_N} (\mathbf{m}_N-\boldsymbol\mu))} \bigg] = e^{-\sqrt{N}( \mathbf{t},\sqrt{\boldsymbol{\alpha}_N} \boldsymbol\mu)}\int_{\mathbb{R}^K} e^{\sqrt{N}( \mathbf{t},\sqrt{\boldsymbol{\alpha}_N} \mathbf{m}_N)} \mathcal{G}_{N}(\mathbf{m}_N) d\mathbf{m}_N \\
= e^{-\sqrt{N}( \mathbf{t},\sqrt{\boldsymbol{\alpha}_N} \boldsymbol\mu)}\frac{Z_{N,\mathbf{t}}}{Z_N} \, .
\end{multline} 

From the last equality on the right of \eqref{mgf} and using \eqref{part_fnx asymp exp}:

\begin{equation}\label{mgf partition}
\begin{split}
 \frac{Z_{N,\mathbf{t}}}{Z_N} \sim &\frac{\exp{(N \max_{\mathbf{x}} f_{N,\mathbf{t}}(\mathbf{x}))}}{\exp{(N \max_{\mathbf{x}} f_N(\mathbf{x}))}}\cr
 =& \exp{\left( N[f_{N}(\boldsymbol\mu_{N,\mathbf{t}})  - f_{N}(\boldsymbol\mu_{N})] + \sqrt{N} \left( \mathbf{t}, \sqrt{\boldsymbol\alpha_N} \boldsymbol\mu_{N,\mathbf{t}}\right) \right) }
\end{split}
\end{equation}
where $\sim$ means equality up to $(1+o(1))$.
In the last equality above, we used that $\boldsymbol\mu_{N,\mathbf{t}}$  and $\boldsymbol\mu_{N}$ are unique global maximizers of $f_{N,\mathbf{t}}(\mathbf{x})$  and $f_{N}(\mathbf{x})$ respectively. Now, following from Lemma \ref{generalized expansion pt}, equation \eqref{mgf partition} becomes:
\begin{equation}
\frac{Z_{N,\mathbf{t}}}{Z_N} \sim \exp{\left( -\frac{1}{2} \left( \mathbf{t},  \sqrt{\boldsymbol\alpha_N}\mathcal{H}_{f_N}^{-1} (\boldsymbol\mu_N) \sqrt{\boldsymbol\alpha_N}  \mathbf{t}  \right) + \sqrt{N} \left( \mathbf{t}, \sqrt{\boldsymbol\alpha_N} \boldsymbol\mu_{N}\right)  \right) } .
\end{equation}

Therefore,

\begin{equation}\label{clt limit expectation}
\begin{split}
\mathbb{E}\bigg[e^{\sqrt{N}( \mathbf{t},\sqrt{\boldsymbol{\alpha}_N} (\mathbf{m}_N-\boldsymbol\mu))} \bigg] \sim & e^{-\sqrt{N} \left( \mathbf{t}, \sqrt{\boldsymbol\alpha_N} \boldsymbol\mu\right)} \cdot e^{\sqrt{N} \left( \mathbf{t}, \sqrt{\boldsymbol\alpha_N} \boldsymbol\mu_{N}\right)} \cdot e^{- \frac{1}{2}\left( \mathbf{t}, \sqrt{\boldsymbol\alpha_{N}}\mathcal{H}^{-1}_{f_N}(\boldsymbol\mu_N) \sqrt{\boldsymbol\alpha_{N}} \mathbf{t}  \right) }\\
= & e^{\sqrt{N}\Big[\mathbf{t} \sqrt{\boldsymbol\alpha_N} \left( \boldsymbol\mu_N - \boldsymbol\mu \right)  \Big]} \cdot e^{- \frac{1}{2}\left( \mathbf{t}, \sqrt{\boldsymbol\alpha_{N}}\mathcal{H}^{-1}_{f_N}(\boldsymbol\mu_N) \sqrt{\boldsymbol\alpha_{N}} \mathbf{t}  \right) }.
\end{split}
\end{equation}
Now, from \eqref{f_N} we have that
\begin{equation}\label{mu tanh vec}
\boldsymbol\mu_N = \tanh{\left(  \mathbf{J} \boldsymbol\alpha_N \boldsymbol\mu_N + \mathbf{h} \right) } \,.
\end{equation}

Let $\boldsymbol{\beta}=(\beta_p)_{p\leq K}$  with  $|\beta_p|<\infty$, and $\theta \in \left[\frac{1}{2}, \infty\right) $. Assume that $\boldsymbol{\alpha}_N\equiv\boldsymbol{\alpha}(\boldsymbol{\beta})= \boldsymbol{\alpha}+ N^{-\theta}\text{diag}(\boldsymbol{\beta})$. If we set $\mathbf{b}_N = \mathbf{J} \boldsymbol\alpha_N \boldsymbol\mu_N + \mathbf{h}$, then 
\begin{equation}
\frac{\partial \boldsymbol{\mu}_{N}}{\partial \boldsymbol{\beta}} = \frac{\partial \boldsymbol{\mu}_{N}}{\partial \mathbf{b}_N} \cdot \frac{\partial \mathbf{b}_N}{\partial \boldsymbol{\beta}}
\end{equation}
where 
\begin{equation}
\frac{\partial \boldsymbol{\mu}_{N}}{\partial \mathbf{b}_N}= \text{diag}(1-\mu_{N,l}
^2)_{l\leq K} \quad \text{and} \quad \frac{\partial \mathbf{b}_N}{\partial \boldsymbol{\beta}} = \mathbf{J}\, \frac{\partial \boldsymbol\alpha_{N}}{\partial \boldsymbol{\beta}} \boldsymbol{\mu}_{N} + \mathbf{J}\,\boldsymbol{\alpha}_N \frac{\partial \boldsymbol{\mu}_{N}}{\partial \boldsymbol{\beta}} \;.%\quad \text{and} \quad \frac{\partial \boldsymbol\alpha_{N}}{\partial \boldsymbol{\beta}} = N^{-\theta} \text{diag}(\mathbf{1}).
\end{equation}
Let, $\mathbf{M} = \text{diag}(1-\mu_{N,l}
^2)_{l\leq K}$ then
\begin{equation} \label{derv mu_N beta}
\frac{\partial \boldsymbol{\mu}_{N}}{\partial \boldsymbol{\beta}} = \mathbf{M} \left[\mathbf{J}\, \frac{\partial \boldsymbol\alpha_{N}}{\partial \boldsymbol{\beta}} \boldsymbol{\mu}_{N} + \mathbf{J}\,\boldsymbol{\alpha}_N \frac{\partial \boldsymbol{\mu}_{N}}{\partial \boldsymbol{\beta}}\right]
\end{equation}which entails
\begin{equation}
\frac{\partial \boldsymbol{\mu}_{N}}{\partial \boldsymbol{\beta}} \Big[\mathbf{M}^{-1} - \mathbf{J}\,\boldsymbol{\alpha}_N \Big] = N^{-\theta} \, \mathbf{J}\, \text{diag}(\boldsymbol{\mu}_{N}).
\end{equation}
Since $\mathcal{H}_{f_N}(\boldsymbol\mu_N) = - \Big[\mathbf{M}^{-1} - \mathbf{J}\,\boldsymbol{\alpha}_N \Big]\cdot \boldsymbol{\alpha}_N$, hence, \eqref{derv mu_N beta} yields:
\begin{equation}
\begin{split}
\frac{\partial \boldsymbol{\mu}_{N}}{\partial \boldsymbol{\beta}} =& \Big[\mathbf{M}^{-1}-\mathbf{J}\boldsymbol{\alpha}_N  \Big]^{-1} N^{-\theta} \, \mathbf{J}\, \text{diag}(\boldsymbol{\mu}_{N})= -N^{-\theta}\, \mathcal{H}_{f_N}^{-1}(\boldsymbol\mu_N) \, \boldsymbol\alpha_N \, \mathbf{J}\, \text{diag}(\boldsymbol{\mu}_{N}).
\end{split}
\end{equation}
Finally, a Taylor's expansion of $\boldsymbol\mu_N$ around $\boldsymbol{\beta}=\mathbf{0}$ gives:
\begin{equation}\label{taylor mu_n}
\boldsymbol{\mu}_N - \boldsymbol{\mu} = -N^{-\theta}\,    \mathcal{H}_{f}^{-1}(\boldsymbol\mu) \,\boldsymbol\alpha \, \mathbf{J}\, \text{diag}(\boldsymbol{\mu})\boldsymbol{\beta} +  \mathcal{O}\left( {N^{-2\theta}} \right) .
\end{equation}
Therefore, Theorem \ref{uniqness CLT} follows from \eqref{taylor mu_n} and taking the limit as $N\to \infty$ of \eqref{clt limit expectation}.

\subsection{Proof of Theorem \ref{nonuniqness CLT}}
%%%%%%
We will now address the case where $f$ reaches its maximum in more than one point.

\begin{lemma}\label{lemma ZN many max}
Suppose $f(\mathbf{x})$ has $n$ global maximizers $\boldsymbol{\mu}^i$ for $i\leq n$ such that $\mathcal{H}_{f}({\boldsymbol\mu}^i)\prec 0$.
For each $i \leq n$  let $A_i\subset [-1,1]^K$  be a poly-interval such that $\boldsymbol{\mu}^i\in \text{int}(A_i)$ is the unique maximizer of $f$ on $\text{cl}(A_i)$. Then, for each $i\leq n$ and for $N$ large enough, $f_{N,\mathbf{t}}$ has a unique global maximizer ${\boldsymbol\mu}_{N,\mathbf{t}}^i \to \boldsymbol\mu$
 on $A_i$ with $\mathcal{H}_{f_{N,\mathbf{t}}}({\boldsymbol\mu}_{N,\mathbf{t}}^i)<0$.  Moreover  for $\delta\in\left( 0, \frac{1}{2K+4}\right) $ one has
\begin{equation}\label{conc2}
\mathcal{G}_{N,\mathbf{t}}(\mathbf{m}_N\in B_{N,\delta}^c({\boldsymbol\mu}_{N,\mathbf{t}}^i) |\mathbf{m}_N \in A_i)= \exp\bigg\{\frac{1}{2}N^{2\delta}\lambda^i_{N,\mathbf{t}}\bigg\}\mathcal{O}\left( N^{\frac{3K}{2}}\right) 
\end{equation}
where $\lambda^i_{N,\mathbf{t}}<0$ is the largest eigenvalue of $\mathcal{H}_{{f}_{N,\mathbf{t}}} (\boldsymbol\mu_{N,\mathbf{t}}^i)$ for ${\boldsymbol\mu}_{N,\mathbf{t}}^i \in int (A_i)$ and 

\begin{equation}\label{conc3}
\mathcal{G}_{N,\mathbf{t}}\left( \mathbf{m}_N\in B_{N,\delta, n}^c \right) = \exp\bigg\{\frac{1}{2}N^{2\delta} \max_{1\leq i\leq n}\lambda^i_{N,\mathbf{t}}\bigg\}\mathcal{O}\left( N^{\frac{3K}{2}}\right) 
\end{equation}
where $B_{N,\delta, n}^c = \bigcup_{i\leq n} B_{N,\delta}^c({\boldsymbol\mu}_{N,\mathbf{t}}^i). $
The partition function restricted to the interval $A_i$, can be expanded as: \begin{equation}\label{z mi expansion}
\begin{split}
    Z_{N,\mathbf{t}}\big|_{A_i}&= \sum_{\mathbf{x}\in S_N\cap A_i} \prod_{l=1}^K A_{N_l}(x_l) \exp{N(f_{N,\mathbf{t}}(\mathbf{x}))} \\
    &= \dfrac{e^{Nf_{N,\mathbf{t}}(\boldsymbol\mu_{N,\mathbf{t}}^i)}}{\sqrt{\det{\left(-\mathcal{H}_{{f}_{N,\mathbf{t}}} (\boldsymbol\mu_{N,\mathbf{t}}^i)\right)} \prod_{l=1}^K(1-(\mu_{N,\mathbf{t}}^{i,(l)})^2)}}\cdot \left(  1+O\left( N^{-1/2+(K+2)\delta} \right)  \right) .
\end{split}
\end{equation}
\end{lemma}

\begin{proof}
For $N$ large enough, $B_{N,\delta}({\boldsymbol\mu}_{N,\mathbf{t}}^i)\subset A_i$ for all $i\leq n$ and equation 
\eqref{conc2} is obtained following a step-by-step argument used to prove equation \eqref{conc1}. Hence, it follows that,

\begin{equation}\label{result conc2}
\mathcal{G}_{N,\mathbf{t}}(\mathbf{m}_N\in B_{N,\delta}^c({\boldsymbol\mu}_{N,\mathbf{t}}^i) |\mathbf{m}_N \in A_i)= \exp\bigg\{\frac{1}{2}N^{2\delta}\lambda^i_{N,\mathbf{t}}\bigg\}\mathcal{O}\left( N^{\frac{3K}{2}}\right) .
\end{equation}

Now, let's observe that, for $N$ large enough, $A_i\setminus B_{N,\delta}({\boldsymbol\mu}_{N,\mathbf{t}}^i) = A_i\setminus B_{N,\delta,n}$ for $i\leq n$. Hence, for all $1\leq i \leq n$ and $N$ large, it follows that 
\begin{equation}
\mathcal{G}_{N,\mathbf{t}}(\mathbf{m}_N\in B_{N,\delta}^c({\boldsymbol\mu}_{N,\mathbf{t}}^i) \vert \mathbf{m}_N \in A_i)= \mathcal{G}_{N,\mathbf{t}}\left( \mathbf{m}_N\in B_{N,\delta, n}^c \vert \mathbf{m}_N \in A_i\right) . 
\end{equation}
Therefore, we have that,
\begin{equation}\label{result conc3}
\begin{split}
\mathcal{G}_{N,\mathbf{t}}\left( \mathbf{m}_N\in B_{N,\delta, n}^c\right)  =& \sum_{1\leq i\leq n} \mathcal{G}_{N,\mathbf{t}}\left( \mathbf{m}_N\in B_{N,\delta, n}^c \vert \mathbf{m}_N \in A_i\right)  \mathcal{G}_{N,\mathbf{t}}\left( \mathbf{m}_N\in A_i\right)  \cr
\leq & \exp\bigg\{\frac{1}{2}N^{2\delta}\max_{1\leq i \leq n}\lambda^i_{N,\mathbf{t}}\bigg\}\mathcal{O}\left( N^{\frac{3K}{2}}\right)  \sum_{1\leq i\leq n}   \mathcal{G}_{N,\mathbf{t}}\left( \mathbf{m}_N\in A_i\right) \cr
=& \exp\bigg\{\frac{1}{2}N^{2\delta}\max_{1\leq i \leq n}\lambda^i_{N,\mathbf{t}}\bigg\}\mathcal{O}\left( N^{\frac{3K}{2}}\right) .
\end{split}
\end{equation}
This complete the result in \eqref{conc3} following from \eqref{result conc3}.

The proof for the asymptotic expansion of the partition function when there are multiple vectors of global maximizers of $f_{N,\mathbf{t}}$ follows exactly the same argument for the case with a unique vector of global maximizer when conditioned on an interval containing only one of the global maximizers. 
Notice that for fixed $i\leq n$ and $N$ large,  $\mathbf{m}_N$ concentrates around ${\boldsymbol\mu}_{N,\mathbf{t}}^i \in A_i$ as it is stated in equation \eqref{conc2}. Hence,
\begin{equation}
    \mathcal{G}_{N,\mathbf{t}}(\mathbf{m}_N\in B_{N,\delta}({\boldsymbol\mu}_{N,\mathbf{t}}^i) |\mathbf{m}_N \in A_i) = \dfrac{1}{Z_{N,\mathbf{t}}\big|_{A_i}} \sum_{\mathbf{x}\in S_N\cap B_{N,\delta}} \prod_{l=1}^K 
\binom{N_l}{\frac{N_l(1+x_l)}{2}} \exp\big\{-H_{N,\mathbf{t}}(\mathbf{x})\big\}.
\end{equation}
Now, following the exact computation and argument in the uniqueness case, the restricted partition function for each of the global maximizers $x_i$ can be expanded as
\begin{equation}\label{z mi}
    Z_{N,\mathbf{t}}\big|_{A_i}= \dfrac{e^{Nf_{N,\mathbf{t}}(\boldsymbol\mu_{N,\mathbf{t}}^i)}}{\sqrt{\det{\left(-\mathcal{H}_{{f}_{N,\mathbf{t}}} (\boldsymbol\mu_{N,\mathbf{t}}^i)\right)} \prod_{l=1}^K(1-((\mu_{N,\mathbf{t}}^{(l)})^{i})^2)}}\cdot \left(  1+\mathcal{O}\left( N^{-1/2+(K+2)\delta} \right)  \right) .
\end{equation}
This completes the proof of Lemma \ref{lemma ZN many max}.
\end{proof}

The conditional moment generating function for some parameter $\mathbf{t}\in \mathbb{R}^K$ can be computed as:
\begin{multline}\label{mgf 2}
\mathbb{E}\bigg[e^{\sqrt{N}( \mathbf{t},\sqrt{\boldsymbol{\alpha}_N} (\mathbf{m}_N-\boldsymbol\mu ^i))}\Big|\big\{\mathbf{m}_N \in A_i\big\} \bigg] = e^{-\sqrt{N} \left( \mathbf{t}, \sqrt{\boldsymbol\alpha_N} \boldsymbol\mu ^i\right)}\int_{\mathbb{R}^K} e^{-\sqrt{N} \left( \mathbf{t}, \sqrt{\boldsymbol\alpha_N} \mathbf{m}_N\right)} \mathcal{G}_{N}\left(\mathbf{m}_N\right) |_{\left\{\mathbf{m}_N \in A_i\right\}}d\mathbf{m}_N \\
= e^{-\sqrt{N} \left( \mathbf{t}, \sqrt{\boldsymbol\alpha_N} \boldsymbol\mu ^i\right)}\frac{Z_{{N,\mathbf{t}}|A_i}}{Z_{N|A_i}}\,.
\end{multline} 

Using the asymptotic expansion of the perturbed partition function in \eqref{z mi}, the proof follows the same arguments as in Theorem \ref{uniqness CLT}.
Hence,

\begin{equation}\label{mgf partition 2}
 \frac{Z_{{N,\mathbf{t}}|A_i}}{Z_{N|A_1}} \sim  \exp{\left( N[f_{N}(\boldsymbol\mu ^i_{N,\mathbf{t}})  - f_{N}(\boldsymbol\mu ^i_{N})] + \sqrt{N} \left( \mathbf{t}, \sqrt{\boldsymbol\alpha_N} \boldsymbol\mu ^i_{N,\mathbf{t}}\right) \right) }
\end{equation}
where $\boldsymbol\mu ^i_{N,\mathbf{t}} \in int(A_i)$ is the unique maximizer of $f_{N,\mathbf{t}}$ and $\boldsymbol\mu ^i_{N}\in int(A_i)$ is the unique maximizer of $f_{N}$. Now following the argument of Lemma \ref{generalized expansion pt} and the proof of Theorem \ref{uniqness CLT}, we have that:
\begin{multline}\label{clt limit expectation non uniq}
\mathbb{E}\Big[e^{\sqrt{N}( \mathbf{t},\sqrt{\boldsymbol{\alpha}_N} (\mathbf{m}_N-\boldsymbol\mu^i))}\Big|\big\{\mathbf{m}_N \in A_i\big\} \Big] \sim  e^{-\sqrt{N} \left( \mathbf{t}, \sqrt{\boldsymbol\alpha_N} \boldsymbol\mu^i\right)} \cdot e^{\sqrt{N} \left( \mathbf{t}, \sqrt{\boldsymbol\alpha_N} \boldsymbol\mu_{N}^i\right)} \cdot e^{- \frac{1}{2}\left( \mathbf{t}, \sqrt{\boldsymbol\alpha_{N}}\mathcal{H}^{-1}_{f_N}(\boldsymbol\mu_N^i) \sqrt{\boldsymbol\alpha_{N}} \mathbf{t}  \right) }\\
= e^{\sqrt{N}\Big[\mathbf{t} \sqrt{\boldsymbol\alpha_N} \left( \boldsymbol\mu_N ^i - \boldsymbol\mu ^i \right) . \Big]} \cdot e^{- \frac{1}{2}\left( \mathbf{t}, \sqrt{\boldsymbol\alpha_{N}}\mathcal{H}^{-1}_{f_N}(\boldsymbol\mu_N^i) \sqrt{\boldsymbol\alpha_{N}} \mathbf{t}  \right) }.
\end{multline}

Similarly, if we set $\boldsymbol{\alpha}_N\equiv\boldsymbol{\alpha}(\boldsymbol{\beta})= \boldsymbol{\alpha}+ N^{-\theta}\text{diag}(\boldsymbol{\beta})$, it follows from \eqref{taylor mu_n} that:
\begin{equation}\label{taylor mu_nl nonuniq}
\boldsymbol{\mu}_N^i - \boldsymbol{\mu}^i = -N^{-\theta}\,    \mathcal{H}_{f}^{-1}(\boldsymbol\mu ^i) \,\boldsymbol\alpha \, \mathbf{J}\, \text{diag}(\boldsymbol{\mu}^i)\boldsymbol{\beta} +  \mathcal{O}\left( {N^{-2\theta}} \right) 
\end{equation}
for $\boldsymbol{\beta}=(\beta_p)_{p\leq K}$  where  $0<\beta_p<\infty$ and $\theta \in \left[\frac{1}{2}, \infty\right) $. Hence, Theorem \ref{nonuniqness CLT} follows from \eqref{taylor mu_nl nonuniq} by taking the limit as $N\to \infty$ of \eqref{clt limit expectation non uniq}. This completes the proof of Theorem \ref{nonuniqness CLT}.

\section{Conclusion}
In this work, we introduced a generalized version of the multispecies Curie-Weiss model featuring arbitrary spins and a non-definite interaction matrix. Specifically, we computed the pressure per-particle  and established the validity of CLT for a suitably rescaled vector of global magnetization. Our results demonstrated that the rescaled vector of global magnetization follows either a centered or non-centered multivariate normal distribution, depending on the speed of convergence of the relative densities of particles to their limiting values as the system size approaches infinity.

This generalized framework holds significant promise for applications in systems where spin particles or units can exhibit arbitrary states, moving beyond traditional binary or discrete-valued models. Such flexibility is particularly relevant in fields like statistical physics, network theory, and social dynamics, where elements often have a continuum of possible states or operate under complex interactions. Furthermore, the inclusion of an indefinite interaction matrix enables the modeling of systems with both cooperative and antagonistic interactions, broadening the scope of potential applications to include meta-magnets, financial systems, and multi-agent interactions.

Future work could explore extending these results to models with more complex topologies or additional constraints, such as time-varying interaction matrices or non-equilibrium dynamics \cite{Andreis_Tovazzi_2018, Ayi_2021, Giacomin_Poquet_2015}. Additionally, investigating the robustness of the CLT under perturbations in the interaction matrix or in the distribution of spin states could yield further insights into the stability and universality of the model's behavior and broaden its applicability \cite{Collet_Formentin_Tovazzi_2016}.

\section*{Acknowledgment}
The authors thank Pierluigi Contucci for inspiring this collaboration. E. M. was supported by the EU H2020 ICT48 project Humane AI Net contract number 952026; by the Italian Extended Partnership PE01 - FAIR Future Artificial Intelligence Research - Proposal code PE00000013 under the MUR National Recovery and Resilience Plan. E. M. and G. O. were supported by the grant for the project PRIN22CONTUCCI, 2022B5LF52 “Boltzmann Machines beyond the “Independent Identically Distributed” Paradigm: a Mathematical Physics Approach”, CUP J53D23003690006, under the MUR National Recovery and Resilience Plan.
%%%%%%%%%%%%%%%%%%%%%%%%%%%%%%%%%%%%%
\appendix
\section{Technical tools}

\begin{lemma}\label{maximizer mu}
The maximizer $\boldsymbol{\mu}$ of $f(\mathbf{x})$ belongs to the interior of  $[-1,1]^K$. 
\end{lemma}

\begin{proof}
Observe that the function  $f(\mathbf{x})$ 
satisfies
\begin{equation}\label{derivative_VP}
\frac{\partial f}{\partial x_l} = \sum_{p=1}^K \Delta_{l,p} x_p + \tilde{h}_l - \alpha_l \left(  \frac{1}{2}\log\left( \frac{1+x_l}{1-x_l}\right)  \right)  \quad \text{for} \quad l=1,...,K.    
\end{equation} 

Suppose that by contradiction $\mu_l^2=1$ for some $l\leq K$. Then if we consider the function $f(x_l)=f(\mathbf{x})\Big|_{x_j=\mu_j,j\neq l}$, it follows that $f(x_l)$ attains it maximum in at least one point $x_l \in (-1,1)$ which satisfy \eqref{Kgroup MFE}. Indeed, from \eqref{derivative_VP}, $\lim_{{x_l} \to -1^+}\frac{\partial }{\partial x_l} f({x_l}) = +\infty$ and $\lim_{{x_l} \to 1^-} \frac{\partial }{\partial x_l} f({x_l}) = -\infty$ . Therefore, there exists $\epsilon > 0$ such that $f(x_l)$ is strictly increasing on $[-1, -1 + \epsilon]$ and strictly decreasing on $[1 - \epsilon, 1]$. Since this is true for all $l=1, \ldots, K$, then the thesis follows.
\end{proof}

%%%%%%%%%%%%%%%%
\begin{lemma}\label{uniform convergence}
Let $\boldsymbol{\mu}$ be a point on the interior of $[-1,1]^K$ and let $A$ be an open neighborhood of $\boldsymbol\mu$. Let $f:cl(A)\to \mathbb{R}$
and assume that $\boldsymbol{\mu}$  is the  unique global maximum point of $f$   and the Hessian $\mathcal{H}_{f}(\boldsymbol{\mu})$ is negative definite. Let 
$(f_N)$ be a sequence of functions  with bounded partial derivatives up to order 2 converging uniformly to those of $f$. Then for $N$ large enough, $f_N$ has a unique maximizer $\boldsymbol\mu_N\rightarrow \boldsymbol{\mu}$  and  $\mathcal{H}_{f_N}(\boldsymbol\mu_N)\prec 0$.
\end{lemma}

\begin{proof}
Suppose that $\{\mathbf{x}_N\}$ is a sequence of  any maximizer of $f_N$ which exists since $cl(A)$ is compact.  Then there exists a subsequence $\{N_n\}_{n\geq 1}$  such that $\{\mathbf{x}_{N_n}\}$ converges to some $\mathbf{y}$. Clearly $f_{N_n}(\mathbf{x}_{N_n}) \geq f_{N_n}({\mathbf{x}})$ for all $\mathbf{x}\in cl(A)$.
Therefore by uniform convergence and taking the limit as $ n\rightarrow \infty$,  we obtain that $f(\mathbf{y})\geq f(\mathbf{x})$ for all $\mathbf{x}\in cl(A)$. This implies that $\mathbf{y}=\boldsymbol{\mu}$ by uniqueness  of the global maximizers of $f$ and therefore, $\mathcal{H}_f(\boldsymbol{\mu})\prec 0$.  Now since $f_N$ converges uniformly to $f$ one has that for $N$ large enough the maximizer $\boldsymbol\mu_N$ of $f_N$ is unique and  $\mathcal{H}_{f_N}(\boldsymbol\mu_N)\prec 0$.\\
\end{proof}

\begin{lemma}\label{sup derivative zeta'}
For $\delta\in (0,1/6]$, the following bound holds:
\begin{equation}\label{sup derivative zeta}
\sup_{\mathbf{x}\in B_{N,\delta}} |\nabla\zeta_{N,\mathbf{t}}(\mathbf{x})| \leq \zeta_{N, \mathbf{t}}(\boldsymbol\mu_{N,\mathbf{t}})\mathcal{O}\left( N^{1/2+\delta}\right)  .
\end{equation}
\end{lemma}
\begin{proof}
The proof is carried in two steps, the first is to compute $\nabla\zeta_{N}(\mathbf{x})$ and the second finds the supremum. 

\textbf{Step 1:} Let's recall from \eqref{partition fnx exp} that, 
\begin{equation}
\zeta_{N,\mathbf{t}}(\mathbf{x}) = \prod_{l=1}^K 
\binom{N_l}{\frac{N_l(1+x_l)}{2}} \exp\big\{-H_{N,\mathbf{t}}(\mathbf{x})\big\} .
\end{equation}
Suppose that $\frac{N_l(1+x_l)}{2}$ is any real number, then the binomial coefficient becomes a continuous binomial coefficient and can be expanded using the arguments of \cite{MSB21,Salwinski}. Now, using gamma functions, we have that for each $l=1,\ldots, K$:
\begin{equation}
A_{N_l}(x_l) =\binom{N_l}{\frac{N_l(1+x_l)}{2}}= \frac{\Gamma\left( N_l+1\right) }{\Gamma\left( \frac{N_l(1+x_l)}{2}+1\right) \Gamma\left( \frac{N_l(1-x_l)}{2}+1\right) }.
\end{equation}
For a given component $l$, differentiating with respect to $x_l$ gives:

\begin{equation}
\frac{\partial A_{N_l}(x_l)}{\partial x_l} = A_{N_l}(x_l) \left(  -\psi\left( \frac{N_l(1+x_l)}{2}+1\right)  \cdot \frac{N_l}{2} + \psi\left( \frac{N_l(1-x_l)}{2}+1\right)  \cdot \frac{N_l}{2} \right) .
\end{equation}

Here, $\psi(z)$ is the digamma function, the derivative of $\log \Gamma(z)$. Now, using the asymptotic expansion of $\psi$  and the properties of $\Gamma$, we have that
\begin{equation}
\begin{split}
\psi\left( \frac{N_l(1+x_l)}{2}+1\right)  =& \log\left( \frac{N_l(1+x_l)}{2}\right)  + \frac{1}{N_l(1+x_l)} + \mathcal{O}\left(  N_l^{-2}\right)  \quad \text{and} \cr
\psi\left( \frac{N_l(1-x_l)}{2}+1\right)  =& \log\left( \frac{N_l(1-x_l)}{2}\right)  + \frac{1}{N_l(1-x_l)} + \mathcal{O}\left(  N_l^{-2}\right) .
\end{split}
\end{equation}

Now, by the product and chain rule, we have that:
\begin{multline}\label{partial zeta}
\frac{\partial \zeta_{N,\mathbf{t}} (\mathbf{x})}{\partial x_l} = \frac{\partial A_{N_l}(x_l)}{\partial x_l}\; 
\prod_{p\neq l} A_{N_p}(x_p)
\cdot \exp{\left(  -H_{N,\mathbf{t}}(\mathbf{x}) \right) } +  A_{N_l}(x_l) \; \cdot \prod_{p\neq l} A_{N_p}(x_p)
\cdot \frac{\partial}{\partial x_l} \exp{\left(  -H_{N,\mathbf{t}}(\mathbf{x}) \right) }\\
= \prod_{p\le K} A_{N_p}(x_p)\frac{N_l}{2}\left(  \log\left( \frac{N_l(1-x_l)}{2}\right)  -\log\left( \frac{N_l(1+x_l)}{2}\right)   + \frac{1}{N_l(1-x_l)} -  \frac{1}{N_l(1+x_l)} + \mathcal{O}\left(  N_l^{-2}\right) \right)  \times\\ 
\times \exp{\left(  -H_{N,\mathbf{t}}(\mathbf{x}) \right) } 
+ \prod_{p\leq K} A_{N_p}(x_p) A_{N_l}(x_l) \; \cdot N\Big\{ \sum_{p=1}^K \Delta_{l,p} x_p + \tilde{h}_l + \frac{t_l \sqrt{N_l}}{N}\Big\} \cdot \exp{\left(  -H_{N,\mathbf{t}}(\mathbf{x}) \right) } \\
= \prod_{p\leq K}A_{N_p}(x_p)  \exp{\left(  -H_{N,\mathbf{t}}(\mathbf{x}) \right) } \left(  -N_l \arctanh{(x_l)} + \frac{x_l}{(1-x_l^2)} 
+ \mathcal{O}\left(  N_l^{-1}\right)  \right)  + \\ 
+ \prod_{p\leq K}A_{N_p}(x_p) \exp{\left(  -H_{N,\mathbf{t}}(\mathbf{x}) \right) } \cdot N\Big\{ \sum_{p=1}^K \Delta_{l,p} x_p + \tilde{h}_l + \frac{t_l \sqrt{N_l}}{N}\Big\}\\
= \underbrace{\prod_{p\leq K}A_{N_p}(x_p) \exp{\left(  -H_{N,\mathbf{t}}(\mathbf{x}) \right) }}_{=\zeta_{N,\mathbf{t}}(\mathbf{x})} \times \\
\times \Bigg[ N \; \underbrace{\Big\{ \sum_{p=1}^K \Delta_{l,p} x_p + \tilde{h}_l - \alpha_{N,l} \arctanh{(x_l)} + \frac{t_l\sqrt{N_l}}{N}\Big\}}_{=\frac{\partial f_{N,\mathbf{t}}}{\partial x_l}} + \frac{x_l}{(1-x_l^2)} + \mathcal{O}\left(  N_l^{-1}\right)  \Bigg]\\
= \zeta_{N,\mathbf{t}}(\mathbf{x}) \left( N \frac{\partial f_{N,\mathbf{t}}(\mathbf{x})}{\partial x_l} + \frac{x_l}{(1-x_l^2)} + \mathcal{O}\left(  N^{-1}\right)   \right) .
\end{multline}

\textbf{Step 2:} 
Observe from the last equality in \eqref{partial zeta} that
\begin{equation}\label{pdef sup zeta}
\sup_{\mathbf{x}\in B_{N,\delta}}|\nabla\zeta_{N}(\mathbf{x})| = \sup_{\mathbf{x}\in B_{N,\delta}} \Big\{ \left|\zeta_{N,\mathbf{t}}(\mathbf{x}) \left( N \frac{\partial f_{N,\mathbf{t}}(\mathbf{x})}{\partial x_l} + \frac{x_l}{(1-x_l^2)} + \mathcal{O}\left(  N^{-1}\right)   \right) \right|\Big\}_{l\leq K}.
\end{equation}
Now, for each $l\in\{1,\ldots,K\}$, by the Mean Value Theorem, there exists a point $\mathbf{c}$ in the line segment connecting $[\mathbf{x},\mu_{N,\mathbf{t}}]$  such that:
\begin{equation}
\begin{split}
\frac{\partial f_{N,{\mathbf{t}}}(\mathbf{x})}{\partial x_l} =& \frac{\partial f_{N,{\mathbf{t}}}(\mathbf{\mu}_{N,\mathbf{t}})}{\partial x_l} + \sum_{p=1}^K\frac{\partial^2 f_{N,{\mathbf{t}}}(\mathbf{c})}{\partial x_l\partial x_p} (x_p - \mu_{N,\mathbf{t},p})\cr
=& \sum_{p=1}^K\frac{\partial^2 f_{N,{\mathbf{t}}}(\mathbf{c})}{\partial x_l\partial x_p} (x_p - \mu_{N,\mathbf{t},p}).
\end{split}
\end{equation}
Hence,
\begin{equation}
   \sup_{\mathbf{x} \in B_{N,\delta}} \left|\frac{\partial f_{N,{\mathbf{t}}}(\mathbf{x})}{\partial x_l}\right| \leq \sup_{\mathbf{c} \in B_{N,\delta}} \sum_{p=1}^K\left|\frac{\partial^2 f_{N,{\mathbf{t}}}(\mathbf{c})}{\partial x_l\partial x_p}\right| N^{-\frac{1}{2}+\delta} .
\end{equation}
Therefore, 
\begin{equation}
\sup_{\mathbf{x} \in B_{N,\delta}} |\nabla f_{N,{\mathbf{t}}}(\mathbf{x})| = \mathcal{O}\left(  N^{-1/2+\delta} \right)     
\end{equation}
and using \eqref{stirlings} we have that 
\begin{equation}\label{sup zeta}
\sup_{\mathbf{x}\in B_{N,\delta}}\zeta_{N,\mathbf{t}}(\mathbf{x}) \leq (1+\mathcal{O}(N^{-1}))\zeta_{N,\mathbf{t}}(\boldsymbol\mu_{N,\mathbf{t}}) \sup_{\mathbf{x}\in B_{N,\delta}} \sqrt{\frac{1-\boldsymbol\mu_{N,\mathbf{t}}^2}{1-\mathbf{x}^2}} = \zeta_{N,\mathbf{t}}(\boldsymbol\mu_{N,\mathbf{t}}) \mathcal{O}(1).
\end{equation}
This implies that,
\begin{equation}
\sup_{\mathbf{x}\in B_{N,\delta}}\left|\frac{\partial \zeta_{N,\mathbf{t}} (\mathbf{x})}{\partial x_l} \right| \leq \zeta_{N,\mathbf{t}}(\boldsymbol\mu_{N,\mathbf{t}}) \mathcal{O}\left(  N^{1/2+\delta}\right)  .
\end{equation}
Hence,
\begin{multline}
\sup_{\mathbf{x}\in B_{N,\delta}}|\nabla \zeta_{N}(\mathbf{x})| \leq \max_{l \in \{1, \ldots, K\}} \left\{ \sup_{\mathbf{x}\in B_{N,\delta}}\zeta_{N,\mathbf{t}}(\mathbf{x}) \; \cdot N \cdot \sup_{\mathbf{c} \in B_{N,\delta}} \sum_{p=1}^K\left|\frac{\partial^2 f_{N,{\mathbf{t}}}(\mathbf{c})}{\partial x_l\partial x_p}\right| N^{-\frac{1}{2}+\delta} \right\}\\
= \zeta_{N,\mathbf{t}}(\boldsymbol\mu_{N,\mathbf{t}}) \mathcal{O}\left(  N^{1/2+\delta}\right) .
\end{multline}
\end{proof}
%%%%%%%%%%

\begin{lemma}\label{generalized expansion pt}
Let $\Omega$ be a bounded  open set in $\mathbb{R}^K$. Let $f_N : \Omega \to \mathbb{R}$ be a sequence of functions such that for $N$ large enough it has a unique global maximum point $\boldsymbol\mu_{N}\in\Omega$ and $\mathcal{H}_{f_{N}}(\boldsymbol\mu_{N})\prec 0$ with bounded partial derivatives up to order 3. For any $\mathbf{t} \in \mathbb{R}^K$, consider the  function
\begin{equation}\label{gNt}
g_{N}(\mathbf{t},\mathbf{x}) = f_N (\mathbf{x}) + \frac{1}{\sqrt{N}} ( \mathbf{t}, \sqrt{\boldsymbol\alpha_N} \mathbf{x}).
\end{equation}
Then for sufficiently large $N$, the function $g_N(\mathbf{t},\mathbf{x})$ also has a unique global maximizer  with $\mathcal{H}_{g_N}(\boldsymbol\mu_{N,\mathbf{t}})\prec 0$ and $\boldsymbol\mu_{N,\mathbf{t}}\to \boldsymbol\mu$ as $N \to \infty$, the following expansion holds: 
\begin{equation}\label{asymptoticgenne}
g_N(\mathbf{t},\boldsymbol\mu_{N,\mathbf{t}}) - f_N(\boldsymbol\mu_{N}) = -\frac{1}{2N} ( \mathbf{t}, \sqrt{\boldsymbol\alpha_{N}}\mathcal{H}^{-1}_{f_N}(\boldsymbol\mu_N) \sqrt{\boldsymbol\alpha_{N}} \mathbf{t}  ) + \frac{1}{\sqrt{N}} (\mathbf{t}, \sqrt{\boldsymbol\alpha_N} \boldsymbol\mu_N) + \mathcal{O}\left( N^{-3/2}\right)  .
\end{equation}
\end{lemma}

\begin{proof}
In order to prove \eqref{asymptoticgenne} let's start with
\begin{equation}
\nabla_{\mathbf{x}} g_N (\mathbf{t},\mathbf{x}) = \nabla_{\mathbf{x}} f_N (\mathbf{x}) + \frac{1}{\sqrt{N}} \sqrt{\boldsymbol\alpha_N} \mathbf{t}.
\end{equation}
Now, since $\boldsymbol\mu_{N,\mathbf{t}}$ is a  maximizer of $g_N (\mathbf{t},\mathbf{x})$ then
\begin{align}
\nabla_{\mathbf{x}} g_N (\mathbf{t},\mathbf{x})\big|_{\mathbf{x} = \boldsymbol\mu_{N,\mathbf{t}}} = \mathbf{0} = \nabla_{\mathbf{x}} f_N (\mathbf{\boldsymbol\mu_{N,\mathbf{t}}}) + \frac{1}{\sqrt{N}} \sqrt{\boldsymbol\alpha_N} \mathbf{t}.
\end{align} 
Now,  we take the gradient on both sides obtaining: 
\begin{comment}
consider the first term on the right hand side of the equation above, and denote by $ \nu(\boldsymbol\mu_{N,\mathbf{t}}) = \nabla_{\mathbf{x}} f_N(\boldsymbol\mu_{N,\mathbf{t}}) : \mathbb{R}^K \to \mathbb{R}^K $. Then 
\begin{equation}
\nabla_{\mathbf{t}}\; \nu(\boldsymbol\mu_{N,\mathbf{t}}) = \mathcal{H}_{f_N}(\mathbf{\boldsymbol\mu_{N,\mathbf{t}}}) \nabla_{\mathbf{t}}\boldsymbol\mu_{N,\mathbf{t}} .
\end{equation}
This implies that,

\begin{equation}
\nabla_{\mathbf{t}} \left( \nabla_{\mathbf{x}} f_N (\boldsymbol\mu_{N,\mathbf{t}}) + \frac{1}{\sqrt{N}} \sqrt{\boldsymbol\alpha_N} \mathbf{t}\right)  = \mathcal{H}_{f_N} (\boldsymbol\mu_{N,\mathbf{t}}) \nabla_{\mathbf{t}}\boldsymbol\mu_{N,\mathbf{t}} + \frac{1}{\sqrt{N}} \sqrt{\boldsymbol\alpha_N} = \mathbf{0}
\end{equation}
and hence,
\end{comment}
\begin{equation}\label{derv mu_N}
\nabla_{\mathbf{t}}\boldsymbol\mu_{N,\mathbf{t}} = -\frac{1}{\sqrt{N}} \mathcal{H}_{f_N}^{-1} (\boldsymbol\mu_{N,\mathbf{t}}) \sqrt{\boldsymbol\alpha_N} .
\end{equation}

Let's notice from equation \eqref{gNt}, and using the fact that $\boldsymbol\mu_{N,\mathbf{t}}$ and $\boldsymbol\mu_N$ are the global maximizers of $g_N (\mathbf{t},\mathbf{x})$ and $f_N(\mathbf{x})$ respectively, we have that:
\begin{equation}
g_{N}(\mathbf{t},\boldsymbol\mu_{N,\mathbf{t}}) - f_N(\boldsymbol\mu_N) =  \underbrace{f_N(\boldsymbol\mu_{N,\mathbf{t}}) - f_N(\boldsymbol\mu_N)}_{=\Phi_N} + \frac{1}{\sqrt{N}} (\mathbf{t}, \sqrt{\boldsymbol\alpha_N} \boldsymbol\mu_{N,\mathbf{t}})  .
\end{equation}
By an application of Taylor's expansion of $f_N(\boldsymbol\mu_{N,\mathbf{t}})$ around $\boldsymbol\mu_{N}$:
\begin{align}\label{Phi N}
\Phi_N = \frac{1}{2} \left( \left( \boldsymbol\mu_{N,\mathbf{t}} - \boldsymbol\mu_{N} \right) , \mathcal{H}_{f_N}(\boldsymbol\mu_N) \left( \boldsymbol\mu_{N,\mathbf{t}} - \boldsymbol\mu_{N} \right)  \right) + \mathcal{O}(N^{-3/2}).
\end{align}
Now to compute $\Phi_N $, we need $\boldsymbol\mu_{N,\mathbf{t}} - \boldsymbol\mu_{N}$. To begin with, let's first observe from equation \eqref{gNt} that, when $\mathbf{t}=\mathbf{0}$: $g_{N}(\mathbf{0},\mathbf{x}) = f_N (\mathbf{x})$ and hence  $\boldsymbol\mu_{N,\mathbf{0}} = \boldsymbol\mu_{N}$. Therefore taking Taylor's expansion of  $\boldsymbol\mu_{N,\mathbf{t}}$ around $\mathbf{t}=\mathbf{0}$ and using \eqref{derv mu_N}:
\begin{equation}
\begin{split}
\boldsymbol\mu_{N,\mathbf{t}} - \boldsymbol\mu_{N,\mathbf{0}} =& \nabla_{\mathbf{t}}\boldsymbol\mu_{N,\mathbf{0}}  \mathbf{t} + \mathcal{O}\left( \frac{1}{N}\right)  \cr
=& -\frac{1}{\sqrt{N}} \mathcal{H}_{f_N}^{-1} (\boldsymbol\mu_N) \sqrt{\boldsymbol\alpha_N}  \mathbf{t} + \mathcal{O}\left( \frac{1}{N}\right)  .
\end{split}
\end{equation}
It now follows from \eqref{Phi N} that,

\begin{equation}
\begin{split}
\Phi_N =& \frac{1}{2N} \left( \mathcal{H}_{f_N}^{-1} (\boldsymbol\mu_N) \sqrt{\boldsymbol\alpha_N}  \mathbf{t},  \mathcal{H}_{f_N}(\boldsymbol\mu_N)\mathcal{H}_{f_N}^{-1} (\boldsymbol\mu_N) \sqrt{\boldsymbol\alpha_N}  \mathbf{t}  \right) + \mathcal{O}(N^{-3/2}) \cr
=& \frac{1}{2N} \left( \mathbf{t},  \sqrt{\boldsymbol\alpha_N}\mathcal{H}_{f_N}^{-1} (\boldsymbol\mu_N) \sqrt{\boldsymbol\alpha_N}  \mathbf{t}  \right) + \mathcal{O}(N^{-3/2})
\end{split}
\end{equation}
and therefore,
\begin{equation}
\begin{split}
g_{N}(\mathbf{t},\boldsymbol\mu_{N,\mathbf{t}}) - f_N(\boldsymbol\mu_N) =& \frac{1}{2N} \left( \mathbf{t},  \sqrt{\boldsymbol\alpha_N}\mathcal{H}_{f_N}^{-1} (\boldsymbol\mu_N) \sqrt{\boldsymbol\alpha_N}  \mathbf{t}  \right)  \cr
&+ \frac{1}{\sqrt{N}} \left(\mathbf{t}, \sqrt{\boldsymbol\alpha_N} \left( \boldsymbol\mu_N -\frac{1}{\sqrt{N}} \mathcal{H}_{f_N}^{-1} (\boldsymbol\mu_N) \sqrt{\boldsymbol\alpha_N}  \mathbf{t}\right)  \right) + \mathcal{O}(N^{-3/2}) \cr
=& \frac{1}{2N} \left( \mathbf{t},  \sqrt{\boldsymbol\alpha_N}\mathcal{H}_{f_N}^{-1} (\boldsymbol\mu_N) \sqrt{\boldsymbol\alpha_N}  \mathbf{t}  \right) + \frac{1}{\sqrt{N}} \left( \mathbf{t}, \sqrt{\boldsymbol\alpha_N} \boldsymbol\mu_N\right) \cr
&- \frac{1}{N} \left( \mathbf{t},  \sqrt{\boldsymbol\alpha_N}\mathcal{H}_{f_N}^{-1} (\boldsymbol\mu_N) \sqrt{\boldsymbol\alpha_N}  \mathbf{t}  \right) + \mathcal{O}(N^{-3/2}) \cr
=& -\frac{1}{2N} \left( \mathbf{t}, \sqrt{\boldsymbol\alpha_N}\mathcal{H}_{f_N}^{-1} (\boldsymbol\mu_N) \sqrt{\boldsymbol\alpha_N}  \mathbf{t}  \right) + \frac{1}{\sqrt{N}} \left( \mathbf{t}, \sqrt{\boldsymbol\alpha_N} \boldsymbol\mu_N\right) + \mathcal{O}(N^{-3/2}) .
\end{split}
\end{equation}
This completes the proof.

\end{proof}

%%%%%%%%%
\section{Approximation lemmas}
%\addcontentsline{toc}{section}{Appendix}
In this section, standard mathematical approximations which played a crucial role in the asymptotic expansion of the partition function is given.

\begin{lemma}[Multidimensional Riemann Approximation]\label{appendix riemann}
Let $Q = [a_1, b_1] \times [a_2, b_2] \times \ldots \times [a_K, b_K]$ be a rectangular domain in $\mathbb{R}^K$, and let $P=\{(x_{1,0}, \ldots, x_{K,0}), (x_{1,1}, \ldots, x_{K,1}), \ldots,$ $ (x_{1,n},  \ldots, x_{K,n})\}$ be any partition of $Q$, where $a_i = x_{i,0} < x_{i,1} < \ldots < x_{i,n} = b_i$ for each $i = 1, 2, \ldots, K$. Assume that $g$ has continuous partial derivatives $\frac{\partial g}{\partial x_i}$ on $Q$ for all $i = 1, 2, \ldots, K$. Let $\epsilon_i = \max_{1 \leq j \leq n} (x_{i,j} - x_{i,j-1})$ denote the mesh size of the partition along the $i$-th variable. Then:
\begin{multline}  \label{eq:approximation_Riemann_sum}
%\begin{split}
\left| \int_Q g(\mathbf{x}) \,d\mathbf{x} - \sum_{j_1,j_2,\dots,j_K=1}^{n}  g(\mathbf{c}_{j_1, j_2, \ldots, j_K}) \cdot \prod_{i=1}^{K} (x_{i,j_i} - x_{i,j_i-1}) \right| \leq \\
Kn^{K-1} \max_{i\leq K}(b^{(i)}-a^{(i)})\max_{\boldsymbol{\xi}\in Q}\|\nabla g(\boldsymbol{\xi})\|\prod_{i=1}^K\epsilon_i 
%\end{split}
\end{multline}
where $\mathbf{c}_{j_1, j_2, \ldots, j_r}$ is any point in the $j_1$-th subinterval along the first variable, $j_2$-th subinterval along the second variable, and so on, up to the $j_K$-th subinterval along the $K$-th variable.
\end{lemma}
\begin{proof}
To begin with, we can decompose the integral into summation of integrals over all the poly-intervals $Q_{j_1,j_2,\dots,j_K}$ constituting the mesh grid:
\begin{align}
    \int_Q g(\mathbf{x}) \,d\mathbf{x}=\sum_{j_1,j_2,\dots,j_K=1}^n\int_{Q_{j_1,j_2,\dots,j_K}} g(\mathbf{x})d\mathbf{x} .
\end{align}
    Using the fact the $g$ is continuous and that each poly-interval is compact, we can use the integral mean value theorem. Therefore for any $j_1,j_2,\dots,j_K$ there exists a $\boldsymbol{\tau}_{j_1,j_2,\dots,j_K}$ such that
    \begin{align}
       \int_{Q_{j_1,j_2,\dots,j_K}} g(\mathbf{x})d\mathbf{x}=g(\boldsymbol{\tau}_{j_1,j_2,\dots,j_K}) \prod_{i=1}^K(x_{i,j_i}-x_{i,j_i-1})\,.
    \end{align}
    This allows us to rewrite the l.h.s. of \eqref{eq:approximation_Riemann_sum} as a unique summation:
    \begin{align}
    \sum_{j_1,j_2,\dots,j_K=1}^{n}\left(g(\boldsymbol{\tau}_{j_1,j_2,\dots,j_K})-g(\mathbf{c}_{j_1,j_2,\dots,j_K})\right)\prod_{i=1}^K(x_{i,j_i}-x_{i,j_i-1}).
    \end{align}
    In order to bound its absolute value we use the triangular inequality, Cauchy-Schwartz inequality and the mutlivariate version of Lagrange's mean value theorem:
\begin{align}
\begin{split} \Bigg\vert\sum_{j_1,j_2,\dots,j_K=1}^{n}&\left(g(\boldsymbol{\tau}_{j_1,j_2,\dots,j_K})-g(\mathbf{c}_{j_1,j_2,\dots,j_K})\right)\prod_{i=1}^K(x_{i,j_i}-x_{i,j_i-1})\Bigg\vert\\
&\leq \sum_{j_1,j_2,\dots,j_K=1}^{n}\max_{\boldsymbol{\xi}\in Q_{j_1,j_2,\dots,j_K}}\|\nabla g(\boldsymbol{\xi})\|\|\mathbf{c}_{j_1,j_2,\dots,j_K}-\boldsymbol{\tau}_{j_1,j_2,\dots,j_K}\|\prod_{i=1}^K\Big\vert x_{i,j_i}-x_{i,j_i-1}\Big\vert\\
&\leq \max_{\boldsymbol{\xi}\in Q}\|\nabla g(\boldsymbol{\xi})\|\prod_{i=1}^K\epsilon_i\sum_{j_1,j_2,\dots,j_K=1}^{n}\|\mathbf{c}_{j_1,j_2,\dots,j_K}-\boldsymbol{\tau}_{j_1,j_2,\dots,j_K}\| \, .
\end{split}
\end{align}
Let us briefly focus on the last factor, and let us consider only the summation w.r.t., say, $j_1$:
\begin{align}
\sum_{j_1=1}^n\|\mathbf{c}_{j_1,j_2,\dots,j_K}-& \boldsymbol{\tau}_{j_1,j_2,\dots,j_K}\|\leq \sum_{i=1}^K\sum_{j_1=1}^n \Big\vert c^{(i)}_{j_1,j_2,\dots,j_K}-\tau^{(i)}_{j_1,j_2,\dots,j_K}\Big\vert\nonumber\\
&\quad \quad \leq \sum_{i=1}^K\sum_{j_1=1}^n\Big\vert x_{j_1}^{(i)}-x_{j_1-1}^{(i)} \Big\vert = \sum_{i=1}^K(b^{(i)}-a^{(i)})\leq K\max_{i\leq K}(b^{(i)}-a^{(i)})
\end{align}
where we used the standard $L^2-L^1$ norm inequality and $^{(i)}$ denotes the i-th component. Hence difference between the Riemann sum and the integral is controlled by
    \begin{align}
        Kn^{K-1} \max_{i\leq K}(b^{(i)}-a^{(i)})\max_{\boldsymbol{\xi}\in Q}\|\nabla g(\boldsymbol{\xi})\|\prod_{i=1}^K\epsilon_i \,.
    \end{align}
\end{proof}

\begin{rem}
Notice that if all $\epsilon_i=\frac{b^{(i)}-a^{(i)}}{n}$ then the above is still of order
\begin{align}
    N n^{K-1}\max_{i\leq K}(b^{(i)}-a^{(i)})^{K+1}\frac{1}{n^K}=\mathcal{O}\left( \frac{1}{n}\right) 
\end{align}which means it still vanishes when the decomposition is fine enough ($n\to\infty$), and if the dimension $K$ is not diverging.
\end{rem}

\begin{lemma}[Multivariate Laplace Approximation] \label{appendix laplace}
Let $f_N: Q\subset\mathbb{R}^K \longrightarrow \mathbb{R}$ be a differentiable sequence of function bounded away from the boundary of $Q = \bigtimes_{l=1}^K\Big[\mu_{N,l} + N_l^{-\frac{1}{2} - \delta}, \mu_{N,l} + N_l^{-\frac{1}{2} + \delta}\Big]$, satisfying $\nabla f_N(\boldsymbol\mu_N) = \mathbf{0}$ and $\mathcal{H}_{f_N}(\boldsymbol\mu_N) \prec 0$, such that $f_N(\boldsymbol\mu_N) > f_N(\mathbf{x})$ for all $\mathbf{x} \in Q$. Let $g(\mathbf{x})$ be analytic function in a neighborhood of $\boldsymbol\mu_N$, then for $\delta \in (0, \frac{1}{6})$ the following hold:
\begin{equation}\label{lemma laplace}
\int_{Q} g(\mathbf{x}) e^{Nf_N(\mathbf{x})} d\mathbf{x} =\sqrt{\frac{(2\pi)^K}{\prod_{l=1}^KN_l\det(-\mathcal{H}_{f_N}(\boldsymbol\mu_N))}} g(\boldsymbol\mu_N) e^{Nf_N(\boldsymbol\mu_N)} (1 + \mathcal{O}(N^{-\frac{1}{2} + \delta})).
\end{equation}

Here, $\nabla f_N$ is the gradient vector.
\end{lemma}

\begin{proof}
Observe from the left hand side of equation \eqref{lemma laplace} that:
\begin{multline}
\mathcal{L} = \int_{Q} g(\mathbf{x}) e^{Nf_N(\mathbf{x})} d\mathbf{x} = \int_{\mu_{N,1} - N_1^{-\frac{1}{2} + \delta}}^{\mu_{N,1} + N_1^{-\frac{1}{2} + \delta}} \cdots \int_{\mu_{N,K} - N_K^{-\frac{1}{2} + \delta}}^{\mu_{N,K} + N_K^{-\frac{1}{2} + \delta}} g(x_1,...,x_K) e^{Nf_N(x_1,...,x_K)} dx_1\cdots dx_K\,.
\end{multline}

Let's consider the following change of variables $t_l = \sqrt{N_l} (x_l - \mu_{N,l})$ for $l = 1, ..., K$. Then it follows that: $x_l = \frac{t_l}{\sqrt{N_l}} + \mu_{N,l}$ and $dx_l = \frac{dt_l}{\sqrt{N_l}}$. Now, the bounds on $\mathbf{x}\in Q$ becomes $t_l\in[-N_l^{\delta},N_l^{\delta}]$ and the integral transforms into:

\begin{multline}\label{laplace exp}
\mathcal{L} = \int_{- N_1^{\delta}}^{N_1^{\delta}} \cdots \int_{- N_K^{\delta}}^{N_K^{\delta}} g\bigg(\frac{t_1}{\sqrt{N_1}} + \mu_{N,1},...,\frac{t_K}{\sqrt{N_K}} + \mu_{N,K}\bigg) e^{Nf_N\left(\frac{t_1}{\sqrt{N_1}} + \mu_{N,1},...,\frac{t_K}{\sqrt{N_K}} + \mu_{N,K}\right)} \prod_{l=1}^K \frac{dt_l}{\sqrt{N_l}}.
\end{multline}
By an application of Taylor expansion of $f_N$ and $g$ around the vector $\boldsymbol\mu_{N}$, we have that: 

\begin{multline}\label{T exp Laplace}
e^{Nf_N\left(\frac{t_1}{\sqrt{N_1}} + \mu_{N,1},...,\frac{t_K}{\sqrt{N_K}} + \mu_{N,K}\right)} = e^{Nf_N(\boldsymbol\mu_{N}) + \frac{1}{2}(\mathbf{t}\mathcal{H}_{f_N}(\boldsymbol\mu_{N}), \mathbf{t} )} \left( 1+\mathcal{O}\left( N\left( \frac{N^{\delta}}{\sqrt{N}}\right) ^3\right)  \right)  \quad \text{and}\\
g\bigg(\frac{t_1}{\sqrt{N_1}} + \mu_{N,1},...,\frac{t_K}{\sqrt{N_K}} + \mu_{N,K}\bigg) = g(\boldsymbol\mu_{N}) + \nabla g(\boldsymbol\mu_{N})\bigg( \frac{N^{\delta}}{\sqrt{N}} \bigg) = \\
=g(\boldsymbol\mu_{N}) \bigg(1 + \mathcal{O}\bigg(N^{\delta - 1/2} \bigg) \bigg) .
\end{multline}
Now, following from \eqref{T exp Laplace}, the right side of \eqref{laplace exp} becomes:

\begin{multline}
\mathcal{L} = \prod_{l=1}^K \frac{1}{\sqrt{N_l}} \left( 1 + \mathcal{O}\left( N^{\delta - 1/2} \right)  \right)   \left( 1 + \mathcal{O}\left( N^{3\delta - 1/2} \right)  \right)  g(\boldsymbol\mu_{N}) e^{Nf_N\left(\boldsymbol\mu_{N}\right)} \cdot \\ \cdot \int_{- N_1^{\delta}}^{N_1^{\delta}} \cdots \int_{- N_K^{\delta}}^{N_K^{\delta}} e^{\frac{1}{2}( \mathbf{t}\mathcal{H}_{f_N}(\boldsymbol\mu_{N}), \mathbf{t} )} d\mathbf{t} \\
= \bigg(1 + \mathcal{O}\bigg(N^{3\delta - 1/2} \bigg) \bigg) \sqrt{\frac{(2\pi)^K}{\prod_{l=1}^KN_l\det{\left(-\mathcal{H}_{f_N}(\boldsymbol\mu_{N})\right)} }  } g(\boldsymbol\mu_{N}) e^{Nf_N\left(\boldsymbol\mu_{N}\right)} \, .
\end{multline}
Notice that we have bounded $t_l$ by its limit $N_l^{\delta}$. This completes the proof of Lemma \ref{appendix laplace}. 
\end{proof}

\bibliographystyle{abbrv}
\bibliography{ref}

\begin{thebibliography}{10}

\bibitem{MSKNL}
D.~Alberici, F.~Camilli, P.~Contucci, and E.~Mingione.
\newblock The multi-species mean-field spin-glass on the nishimori line.
\newblock {\em Journal of Statistical Physics}, 182, 01 2021.

\bibitem{DBMNL}
D.~Alberici, F.~Camilli, P.~Contucci, and E.~Mingione.
\newblock The solution of the deep boltzmann machine on the nishimori line.
\newblock {\em Communications in Mathematical Physics}, 387, 10 2021.

\bibitem{Nishiletter}
D.~Alberici, F.~Camilli, P.~Contucci, and E.~Mingione.
\newblock A statistical physics approach to a multi-channel wigner spiked model.
\newblock {\em Europhysics Letters}, 136(4):48001, 2022.

\bibitem{noi_deep2}
D.~Alberici, P.~Contucci, and E.~Mingione.
\newblock {Deep Boltzmann Machines: rigorous results at arbitrary depth}.
\newblock {\em Annales Institut Henri Poincaré (to appear)}, 2021.

\bibitem{Andreis_Tovazzi_2018}
L.~Andreis and D.~Tovazzi.
\newblock Coexistence of stable limit cycles in a generalized {C}urie–{W}eiss model with dissipation.
\newblock {\em Journal of statistical physics}, 173(1):163–181, 2018.

\bibitem{Ayi_2021}
N.~Ayi and N.~Pouradier~Duteil.
\newblock Mean-field and graph limits for collective dynamics models with time-varying weights.
\newblock {\em Journal of differential equations}, 299:65–110, 2021.

\bibitem{baik2020}
J.~Baik and J.~O. Lee.
\newblock Free energy of bipartite spherical sherrington–kirkpatrick model.
\newblock {\em Annales Institut Henri Poincaré}, 56, 2020.

\bibitem{MSK_original}
A.~Barra, P.~Contucci, E.~Mingione, and D.~Tantari.
\newblock Multi-species mean field spin glasses. rigorous results.
\newblock In {\em Annales Henri Poincar{\'e}}, volume~16, pages 691--708. Springer, 2015.

\bibitem{Barra_Genovese_Guerra_2011}
A.~Barra, G.~Genovese, and F.~Guerra.
\newblock Equilibrium statistical mechanics of bipartite spin systems.
\newblock {\em Journal of physics. A, Mathematical and theoretical}, 44(24):245002, 2011.

\bibitem{Bidaux1986}
R.~Bidaux, N.~Boccara, and G.~Forg{\`a}cs.
\newblock Three-spin interaction {I}sing model with a nondegenerate ground state at zero applied field.
\newblock {\em Journal of statistical physics}, 45:113--134, 1986.

\bibitem{Blume_Durlauf_2003}
L.~Blume and S.~Durlauf.
\newblock Equilibrium concepts for social interaction models.
\newblock {\em International game theory review}, 05(03):193–209, 2003.

\bibitem{Brock_Durlauf_2001}
W.~A. Brock and S.~N. Durlauf.
\newblock Discrete choice with social interactions.
\newblock {\em The Review of Economic Studies}, 68(2):235–260, 2001.

\bibitem{BurioniContucciFedeleVerniaVezzani2015}
R.~Burioni, P.~Contucci, M.~Fedele, C.~Vernia, and A.~Vezzani.
\newblock Enhancing participation to health screening campaigns by group interactions.
\newblock {\em Scientific reports}, 5(1):9904, 2015.

\bibitem{camilli2023central}
F.~Camilli, P.~Contucci, and E.~Mingione.
\newblock Central limit theorem for the overlaps on the nishimori line.
\newblock {\em arXiv preprint arXiv:2305.19943}, 2023.

\bibitem{Collet_2014}
F.~Collet.
\newblock Macroscopic limit of a bipartite {C}urie–{W}eiss model: A dynamical approach.
\newblock {\em Journal of statistical physics}, 157(6):1301–1319, 2014.

\bibitem{Collet_Formentin_Tovazzi_2016}
F.~Collet, M.~Formentin, and D.~Tovazzi.
\newblock Rhythmic behavior in a two-population mean-field {I}sing model.
\newblock {\em Physical review. E}, 94(4–1):042139, 2016.

\bibitem{Contucci_Ghirlanda_2007}
P.~Contucci and S.~Ghirlanda.
\newblock Modelling society with statistical mechanics: an application to cultural contact and immigration.
\newblock {\em Quality and Quantity}, 41:569–578, 2007.

\bibitem{ContucciKO2022}
P.~Contucci, J.~Kert{\'e}sz, and G.~Osabutey.
\newblock Human-{AI} ecosystem with abrupt changes as a function of the composition.
\newblock {\em PloS one}, 17(5):e0267310, 2022.

\bibitem{Contucci_Mingione_Osabutey_2024}
P.~Contucci, E.~Mingione, and G.~Osabutey.
\newblock Limit theorems for the cubic mean-field {I}sing model.
\newblock {\em Annales Henri Poincare. A Journal of Theoretical and Mathematical Physics}, 25(11):5019–5044, 2024.

\bibitem{Dembo_Zeitouni_2010}
A.~Dembo and O.~Zeitouni.
\newblock {\em Large Deviations Techniques and Applications}.
\newblock Springer, Berlin, Germany, 2010.

\bibitem{Dey}
P.~Dey and Q.~Wu.
\newblock Fluctuation results for multi-species sherrington-kirkpatrick model in the replica symmetric regime.
\newblock {\em Journal of Statistical Physics}, 185, 12 2021.

\bibitem{Eisele_Ellis_1988}
T.~Eisele and R.~S. Ellis.
\newblock Multiple phase transitions in the generalized {C}urie-{W}eiss model.
\newblock {\em Journal of statistical physics}, 52(1–2):161–202, 1988.

\bibitem{Ellis_Newman_1978}
R.~S. Ellis and C.~M. Newman.
\newblock The statistics of {C}urie-{W}eiss models.
\newblock {\em Journal of Statistical Physics}, 19(2):149–161, 1978.

\bibitem{Fedele_Contucci_2011}
M.~Fedele and P.~Contucci.
\newblock Scaling limits for multi-species statistical mechanics mean-field models.
\newblock {\em Journal of statistical physics}, 144(6):1186–1205, 2011.

\bibitem{Fleermann_Kirsch_Toth_2022}
M.~Fleermann, W.~Kirsch, and G.~Toth.
\newblock Local central limit theorem for multi-group {C}urie–{W}eiss models.
\newblock {\em Journal of theoretical probability}, 35(3):2009–2019, 2022.

\bibitem{FroyenSH1976}
S.~Fr{\o}yen, A.~S. Sudb{\o}, and P.~Hemmer.
\newblock Ising models with two-and three-spin interactions: mean field equation of state.
\newblock {\em Physica A: Statistical Mechanics and its Applications}, 85(2):399--408, 1976.

\bibitem{GalamYS1998}
S.~Galam, C.~S.~O. Yokoi, and S.~R. Salinas.
\newblock Metamagnets in uniform and random fields.
\newblock {\em Phys. Rev. B}, 57:8370--8374, Apr 1998.

\bibitem{Gallo2008}
I.~Gallo and P.~Contucci.
\newblock Bipartite mean field spin systems. existence and solution.
\newblock {\em Mathematical Physics Electronic Journal [electronic only]}, 14:Paper No. 1, 21 p.--Paper No. 1, 21 p., 2008.

\bibitem{Genovese_Barra_2009}
G.~Genovese and A.~Barra.
\newblock A certain class of {C}urie-{W}eiss models.
\newblock 2009.

\bibitem{G_Tantari_2016}
G.~Genovese and D.~Tantari.
\newblock Non-convex multipartite ferromagnets.
\newblock {\em J. Stat. Phys.}, 163(3):492--513, 2016.

\bibitem{Giacomin_Poquet_2015}
G.~Giacomin and C.~Poquet.
\newblock Noise, interaction, nonlinear dynamics and the origin of rhythmic behaviors. brazilian jounal of probab.
\newblock {\em Stat}, 29(2):460–493, 2015.

\bibitem{Gsänger_Hösel_MKM_2024}
M.~Gsänger, V.~Hösel, C.~Mohamad-Klotzbach, and J.~Müller.
\newblock Opinion models, election data, and political theory.
\newblock {\em Entropy (Basel, Switzerland)}, 26(3):212, 2024.

\bibitem{guerra2002central}
F.~Guerra and F.~Lucio~Toninelli.
\newblock Central limit theorem for fluctuations in the high temperature region of the sherrington--kirkpatrick spin glass model.
\newblock {\em Journal of Mathematical Physics}, 43(12):6224--6237, 2002.

\bibitem{guerra2002quadratic}
F.~Guerra and F.~L. Toninelli.
\newblock Quadratic replica coupling in the sherrington--kirkpatrick mean field spin glass model.
\newblock {\em Journal of Mathematical Physics}, 43(7):3704--3716, 2002.

\bibitem{Bimourrat}
J.-C.~M. Hong-Bin~Chen.
\newblock On the free energy of vector spin glasses with non-convex interactions.
\newblock {\em Preprint arXiv:2311.08980}, 2023.

\bibitem{Kac}
M.~Kac.
\newblock Mathematical mechanisms of phase transitions.
\newblock {\em pp 241-305 of Statistical Physics, Phase Transitions, and Superfluidity. Vol. I. Chretien, M. Gross, E. P. Deser, S. (eds.). New York, Gordon and Breach, Science Publishers, 1968.}, 10 1969.

\bibitem{KinCohen1975}
J.~M. Kincaid and E.~G.~D. Cohen.
\newblock Phase diagrams of liquid helium mixtures and metamagnets: experiment and mean field theory.
\newblock {\em Physics Reports}, 22(2):57--143, 1975.

\bibitem{Kirsch_2016}
W.~Kirsch.
\newblock {\em A mathematical view on voting and power}, page 251–279.
\newblock EMS Press, Zuerich, Switzerland, 2016.

\bibitem{Kirsch_Langner_2014}
W.~Kirsch and J.~Langner.
\newblock {\em The fate of the square root law for correlated voting}, page 147–158.
\newblock Springer International Publishing, Cham, 2014.

\bibitem{Kirsch_Toth_2020}
W.~Kirsch and G.~Toth.
\newblock Two groups in a {C}urie-{W}eiss model.
\newblock {\em Mathematical Physics Analysis and Geometry}, 23(2), 2020.

\bibitem{Kirsch_Toth_2022}
W.~Kirsch and G.~Toth.
\newblock Limit theorems for multi-group {C}urie–{W}eiss models via the method of moments.
\newblock {\em Mathematical Physics Analysis and Geometry}, 25(4), 2022.

\bibitem{Knöpfel_Löwe_Schubert_Sinulis_2020}
H.~Knöpfel, M.~Löwe, K.~Schubert, and A.~Sinulis.
\newblock Fluctuation results for general block spin {I}sing models.
\newblock {\em Journal of statistical physics}, 178(5):1175–1200, 2020.

\bibitem{Löwe_Schubert_2018}
M.~Löwe and K.~Schubert.
\newblock Fluctuations for block spin {I}sing models.
\newblock {\em Electronic communications in probability}, 23, 2018.

\bibitem{Marsman_Tanis_BW_2019}
M.~Marsman, C.~C. Tanis, T.~M. Bechger, and L.~J. Waldorp.
\newblock {\em Network psychometrics in educational practice: Maximum likelihood estimation of the {C}urie-{W}eiss model}, page 93–120.
\newblock Springer International Publishing, Cham, 2019.

\bibitem{mourrat200}
J.-C. {Mourrat}.
\newblock Nonconvex interactions in mean-field spin glasses.
\newblock {\em arXiv e-prints}, 2020.

\bibitem{mukherjee2021variational}
S.~Mukherjee and S.~Sen.
\newblock Variational inference in high-dimensional linear regression, 2021.
\newblock \texttt{arXiv:2104.12232}.

\bibitem{MSB21}
S.~Mukherjee, J.~Son, and B.~B. Bhattacharya.
\newblock Fluctuations of the magnetization in the {$p$}-spin {C}urie-{W}eiss model.
\newblock {\em Comm. Math. Phys.}, 387(2):681--728, 2021.

\bibitem{Opoku_Osabutey_2018}
A.~A. Opoku and G.~Osabutey.
\newblock Multipopulation spin models: A view from large deviations theoretic window.
\newblock {\em Journal of mathematics}, 2018:1–13, 2018.

\bibitem{OpokuOsabuteyK2019}
A.~A. Opoku, G.~Osabutey, and C.~Kwofie.
\newblock Parameter evaluation for a statistical mechanical model for binary choice with social interaction.
\newblock {\em Journal of Probability and Statistics}, 2019, 2019.

\bibitem{MSK_Panchenko}
D.~Panchenko.
\newblock {The free energy in a multi-species Sherrington–Kirkpatrick model}.
\newblock {\em The Annals of Probability}, 43(6):3494 -- 3513, 2015.

\bibitem{Salwinski}
D.~Salwinski.
\newblock The continuous binomial coefficient: An elementary approach.
\newblock {\em The American Mathematical Monthly}, 125(3), 2018.

\bibitem{Subag}
E.~Subag.
\newblock {TAP approach for multispecies spherical spin glasses II: The free energy of the pure models}.
\newblock {\em The Annals of Probability}, 51(3):1004 -- 1024, 2023.

\bibitem{Talagrand2003spin}
M.~Talagrand.
\newblock {\em Spin glasses: a challenge for mathematicians: cavity and mean field models}, volume~46.
\newblock Springer Science \& Business Media, 2003.

\end{thebibliography}
\end{document}